\documentclass[12pt,a4paper]{amsart}
\usepackage{amssymb}
\usepackage{amsfonts}
\usepackage{amsmath}
\numberwithin{equation}{section}
\usepackage[mathscr]{eucal}
\usepackage{enumerate}
\usepackage{color,graphicx,url}

\usepackage{comment}

\renewcommand{\proofname}{\textbf{Proof}}
\makeatletter
\renewenvironment{proof}[1][\proofname]{\par
	\pushQED{\qed}%
	\normalfont \topsep6\p@\@plus6\p@\relax
	\trivlist
	\item\relax
	{\bfseries
		#1\@addpunct{.}}\hspace\labelsep\ignorespaces
}{%
	\popQED\endtrivlist\@endpefalse
}
\makeatother

\allowdisplaybreaks

\newtheorem{theorem}{Theorem}[section]
\newtheorem{corollary}[theorem]{Corollary}

\newtheorem{lemma}[theorem]{Lemma}

\newtheorem{example}[theorem]{Example}

\newtheorem{remark}[theorem]{Remark}

\def\1{\mbox{\boldmath $1$}} 

\def\ii{\mbox{$\mathrm{i}$}}

\def\C{\mathbf{C}}

\def\R{\mathbf{R}}

\def\Z{\mathbf{Z}}

\def\diag{{\mathrm{diag}}}

\def\rank{{\mathrm{rank}}}
\renewcommand{\Re}{\mathrm{Re}}
\renewcommand{\Im}{\mathrm{Im}}
\newcommand{\ZZ}{\mathbf{Z}}
\def\Spec{{\mathrm{Spec}}}

\def\vol{{\mathrm{vol}}}
\def\t{{}^t\!}

\def\wt{{\mathrm{wt}}}
\def\K{{\mathcal{K}}}

\def\cwe{{\mathrm{cwe}}}
\renewcommand{\mod}{\mathrm{mod}}

\title[$I$-Bessel lattice sums]{Lattice sums of $I$-Bessel functions, theta functions, linear codes and heat equations\footnote{Declarations of interest: none.}}
\author{Takehiro Hasegawa}
\address[T. Hasegawa]{Department of Education, Shiga University, Otsu, Shiga 520-0862, JAPAN}
\email{thasegawa3141592@yahoo.co.jp}
\author{Hayato Saigo}
\address[H. Saigo]{Nagahama Institute of Bio-Science and Technology, 1266, Tamura, Nagahama 526-0829, JAPAN}
\email{h\_saigoh@nagahama-i-bio.ac.jp}
\author{Seiken Saito}
\address[S. Saito]{
Division of Liberal Arts, 
Center for Promotion of Higher Education, 
Kogakuin University,
2665-1 Nakano, Hachioji, Tokyo 192-0015, JAPAN
}
\email{saito.seiken@cc.kogakuin.ac.jp}
\author{Shingo Sugiyama}
\address[S. Sugiyama]{Faculty of Mathematics and Physics, Institute of Science and Engineering, Kanazawa University, Kakumamachi, Kanazawa, Ishikawa 920-1192, JAPAN}
\email{s-sugiyama@se.kanazawa-u.ac.jp}


\begin{document}
\begin{abstract}
We extend a certain type of identities on sums of $I$-Bessel functions on lattices, 
previously given by G.~Chinta, J.~Jorgenson, A.~Karlsson and M.~ Neuhauser.
Moreover we prove that, with continuum limit, the transformation formulas of theta functions such as the Dedekind eta function can be given by $I$-Bessel lattice sum identities with characters.
We consider analogues of theta functions of lattices coming from linear codes and 
show that sums of $I$-Bessel functions defined by linear codes can be expressed by complete weight enumerators.    
We also prove that $I$-Bessel lattice sums appear as solutions of heat equations on general lattices.
As a further application, we obtain an explicit solution of the heat equation on $\Z^n$ whose initial condition is given by a linear code.
\end{abstract}

\maketitle

\section{Introduction}

The sum $\sum_{\gamma\in\Gamma} f(\gamma)$ for a function $f$ on a lattice $\Gamma$ is
ubiquitous in mathematics and physics, especially
number theory, lattice field theory, crystal physics, etc.
For example, it is related to theta functions, $L$-functions, Mahler measures, etc.
The Madelung constant of a crystal 
can be expressed by the special value of the corresponding Epstein zeta function
(see \cite{BGMWZ} for details). 

In this article, we consider a lattice sum of the $I$-Bessel functions such as 
\begin{align}\label{latticesum}
	\sum_{\gamma \in \Gamma} \prod_{j=1}^n I_{x_j+\gamma_j}(t_j), \qquad (t_1,\ldots,t_n) \in \C^n.
\end{align}
Here, $I_{x}(t)$ is the $I$-Bessel function or the modified Bessel function of the first kind (\cite{Watson,AAR}), $\Gamma$ is a sublattice of $\Z^n$ and $(x_1,\dots, x_n)\in \Z^n$ is a given vector. 

In \cite{KN,CJK}, it is pointed out that the function $e^{-t}I_{x-y}(t)$ ($x,y\in\Z$) is a discrete analogue of the heat kernel $\frac{1}{\sqrt{4\pi t}}e^{-\frac{|x-y|^2}{4t}}$ ($x,y\in\R$) of the Laplacian of $\R$.
Lattice sums such as  \eqref{latticesum} are important in relation to discrete analogues of theta functions. In fact, we can obtain theta functions by taking a ``continuum limit''
of lattice sums. For example, if $n=1$ and $\Gamma=m\Z$ with a positive integer $m$, then the sum \eqref{latticesum} can be written as follows: 
\begin{align}\label{i-e}
	\sum_{\gamma \in m\Z}  I_{x+\gamma}(t)
	=
	\dfrac{1}{m}
	\sum_{j=0}^{m-1}\exp\left(t\cos \frac{2\pi j}{m}\right)e^{2\pi \ii x \frac{j}{m}}, \qquad (x, t) \in \Z \times \C.
\end{align}
By setting $x=0$, multiplying $me^{-t}$ to the both sides,
replacing $t$ with $2m^2 t$ and taking the limit as $m\to\infty$, we can obtain the transformation formula of the theta function:
\begin{align}\label{thetatrans}
	\dfrac{1}{\sqrt{4\pi t}}\sum_{r\in\Z}e^{-\frac{r^2}{4t}}
	=\sum_{j\in \Z} e^{-4\pi^2 j^2 t}, \qquad t >0.
\end{align}
In this article, we refer to \eqref{thetatrans} as
the \textit{continuum limit} of \eqref{i-e}.
Taking continuum limits
has been discussed in some works \cite{Chung, Chung-Yau97,Chung-Yau2001, KN, Mnev, CJK}.
In \cite{KN}, Karlsson and Neuhauser remarked that the identity \eqref{i-e} could be proved directly by an appropriate (generalized) Poisson summation formula, although they used another method.
In this article, we prove a generalization of the identity of $I$-Bessel lattice sums including \eqref{i-e} for 
any sublattice $\Gamma$ of $\Z^n$ with characters by the Poisson summation formula. The following is a main theorem of this article (see also Theorem \ref{thm:I-E} for trivial characters).

\begin{theorem}\label{thm:I-E:char}
	Set $\chi((a_1,\dots, a_n)):=\prod_{j=1}^n \chi_j(a_j)$, where 
	each $\chi_j$ is a primitive Dirichlet character modulo $q$. 
	Let $\Gamma=\Z^nA$ with $A\in {\rm GL}_n(\R)$ be a sublattice of $\Z^n$. 
	Suppose that every entry of $A$ is divisible by $q$. 
	Then, for any $(t_1,\dots,t_n) \in \C^n$,  any $x=(x_1,\dots, x_n)\in \Z^n$ and any $y=(y_1,\dots, y_n)\in \R^n$,  we obtain 
	\begin{align}\label{chi-I-E}
		&\sum_{\gamma=(\gamma_j)_{j} \in \Gamma} \chi(\gamma A^{-1} )
		\left\{\prod_{j=1}^n I_{x_j+\tfrac{\gamma_j}{q}}(t_j)\right\}e^{2\pi\ii\langle y, x+\tfrac{\gamma}{q}\rangle}\\
		=&
		\dfrac{\prod_{j=1}^n\mathcal{G}(\chi_j)}{\vol(\R^n/\Gamma)}
		\sum_{\gamma^\ast=(\gamma_j^*)_j \in \Gamma^\ast} \bar{\chi} (\gamma^\ast \, \t A)
		\left\{\prod_{j=1}^{n}\1_{[-1/2,1/2]}(y_j-\gamma^\ast_j)e^{t_j\cos 2\pi (y_j- \gamma_j^\ast)}\right\}e^{2\pi \ii\langle x, \gamma^\ast \rangle},
		\notag
	\end{align}
where
each $\chi_j$ is regarded as a function on $\Z$ by
${\chi_j}(a)=0$ if $\gcd(a,q)\neq 1$, 
the sum $\mathcal{G}(\chi_j)=\sum_{a=0}^{q-1}\chi_j(a)e^{2\pi \ii \frac{a}{q}}$ is the Gauss sum of $\chi_j$ if $q>1$ and we put $\mathcal{G}(\chi_j)=1$ if $q=1$,
and ${\bf 1}_{[-1/2,1/2]}$ is the rectangular function given as \eqref{ch}.
\end{theorem}

By taking the continuum limit, we have the following
transformation formula of the Riemann theta functions.
\begin{corollary}
	\label{thm:renormalization}
	Let $\Gamma=\Z^nA$, $q$, $\chi$, $x$ and $y$ be the same as in Theorem \ref{thm:I-E:char}.
	Suppose $(t_1,\dots,t_n)\in \R_{> 0}^n$. Then we have 
	\begin{align}\label{chi-I-E-infty}
		&\frac{1}{(2\pi)^{n/2}}\sum_{\gamma\in \Gamma } \chi(\gamma A^{-1} )
		e^{-\frac{1}{2}\sum_{j=1}^{n}\frac{(x_j+\tfrac{\gamma_j}{q})^2}{t_j}
			+2\pi\ii\langle y, x+\tfrac{\gamma}{q}\rangle}
		\\
		&=
		\dfrac{\prod_{j=1}^n\mathcal{G}(\chi_j)}{\vol(\R^n/\Gamma)}
		\sum_{\gamma^\ast \in \Gamma^\ast}
		\bar{\chi} (\gamma^\ast \, \t A)
		\{\prod_{j=1}^{n}\sqrt{t_j}\}
		e^{-2\pi^2 \sum_{j=1}^n(y_j-\gamma_j^\ast)^2 t_j
			+2\pi \ii\langle x,\gamma^\ast \rangle}.
		\notag
	\end{align} 
\end{corollary}

As an application of our identity \eqref{chi-I-E-infty}, we obtain the transformation formula of 
the Dedekind eta function (see Example \ref{ex:eta}).
If $\chi$ equals the primitive quadratic Dirichlet character $\left(\frac{12}{\cdot }\right)$ modulo $12$,
then we can obtain by Corollary \ref{thm:renormalization} the transformation formula of the Dedekind eta function $\eta$: 
$$
\sqrt{\frac{\ii}{\tau}} \eta(-1/\tau)= \eta(\tau).
$$
Theorem 1.1 is regarded as discrete analogues of
theta inversion formulas
of theta functions of lattices of $\R^n$.
The identity in \cite{CJK} is given only for direct products of $n$-copies of $1$-dimensional lattices.
In contrast, our identity can also be applied for lattices of
non-direct products.

In this article, we also discuss the semidiscrete heat equations for the graph Laplacians of lattices.
For semidiscrete heat equations on $\Z$ and on $\Z^n$, see \cite{GI, KN, CJK, BBDS, Slavik, CJK2, CGRTV}.
As an application of
$I$-Bessel lattice sums,
we give an explicit formula of the heat kernel for the Laplacian on an arbitrary lattice (see Theorem \ref{th:heatLambda}). 
In particular, we give an explicit solution 
of the heat equation on $\Z^n$ whose initial condition is given by a linear code as follows (for the definitions of linear codes, dual codes, complete weight enumerators, see the beginning of Section \ref{Lattice sums of I-Bessel functions coming from linear codes}).

\begin{theorem}\label{thm:heateq-code}
	Let $m\ge 2$ be an integer. 
	Let $C$ be a linear code over $\Z/m\Z$ of length $n$.
	Let $\rho: \Z^n \to (\Z/m\Z)^n$ be the reduction modulo $m$ given as $\rho((a_j)_{j=1}^{n})=(a_j \pmod m)_{j=1}^{n}$. 
	Then the heat equation 
	\begin{align*}
		\begin{cases}
			\partial_t u(x,t) = \Delta_{\Z^n} u(x,t), & \qquad (x, t) \in \Z^n \times \R_{\ge0}, \\
			u(x,0) = u_0(x), & \qquad x \in \Z^n
		\end{cases}
	\end{align*}
has a unique bounded solution $u(x,t)$ $(x\in\Z^n,\ t\in\R_{\ge 0})$ which is differentiable in $t \in \R_{>0}$ and continuous in $t\in \R_{\ge0}$,
		where $\Delta_{\Z^n}$ is the graph Laplacian of $\Z^n$ defined by \eqref{laplacian} and
	$$
	u_0(x)=\1_{\rho^{-1}(C)}(x)=
	\begin{cases}
		1, & (\text{if $x \in \rho^{-1}(C)$}),\\
		0,  & (\text{otherwise}).
	\end{cases}
	$$
The solution $u(x,t)$ is explicitly given by 
	\begin{align}\label{solution}
		u(x,t)
		=e^{-t}\,\frac{\# C}{m^n} \, \sum_{c\in C^\bot}\prod_{j=1}^{n}e^{\tfrac{t}{n}\cos(\frac{2\pi c_j}{m})} e^{\frac{2\pi\ii x_jc_j}{m}},
	\end{align}
where $C^\bot$ denotes the dual code of $C$.
In particular, for any $x \in \rho^{-1}(C)$ we have
\begin{align}\label{solution-c}
	u(x,t)
	&=
	e^{-t}\,\frac{\# C}{m^n} \, \cwe_{C^\bot}(e^{\tfrac{t}{n}\cos \frac{2\pi \cdot 0}{m}}, e^{\tfrac{t}{n}\cos \frac{2\pi \cdot 1}{m}},\dots,e^{\tfrac{t}{n}\cos \frac{2\pi (m-1)}{m}}),
\end{align}
where $\cwe_{C^\bot}$ denotes the complete weight enumerator
of $C^\bot$.
\end{theorem}

Theta functions of lattices of $\R^n$ appear in coding theory (e.g., \cite{NRS, CS, Ebeling}).
In this article, we discuss some kinds of $I$-Bessel lattice sums as discrete analogues of the theta functions of lattices given by
the pull-backs of linear codes over finite rings.
The theta functions of linear codes over a finite ring satisfy
the MacWilliams identity~(\cite{Wood, Nishimura}).
Similarly, our $I$-Bessel lattice sums as discrete analogues of the theta functions of linear codes 
satisfy the identities coming from the MacWilliams identity (see Theorem \ref{thm:code-main}).

This article is concerned with discrete tori and real tori as in
Example \ref{ex:discretetorus} and Remark \ref{trace formula for torus}.
Recently, Xie, Zhao and Zhao \cite{XZZ} gave some identities connecting
spectral zeta functions associated to cycle graphs and Dirichlet $L$-functions
toward the direction similar to ours.
We hope that a discrete analogue of the theta inversion formula in this article
would be useful for the study of those spectral zeta functions and Dirichlet $L$-functions,
as the usual theta inversion formula plays a pivotal role in the analysis of $L$-functions.

This article is organized as follows. 
In Section \ref{Formulas of lattice sums of I-Bessel functions}, we introduce lattices and $I$-Bessel functions,
and discuss the convergence of the $I$-Bessel lattice sums.
Lattice sum identities cannot be proved by using merely the usual Poisson summation formula since our test functions involving $I$-Bessel functions are not continuous.
For the reason, we modify the Poisson summation formula
by dropping the continuity of test functions as in Theorem \ref{PSF}.
Furthermore we also devise computation of our formula in Theorem 
\ref{thm:I-E} by using alternative functions $\tilde I_{x}(t)$ of the $I$-Bessel functions $I_{x}(t)$ since
the Fourier transformation is more tractable.
In Section \ref{Generalizations of sums of the I-Bessel functions on lattices with characters and theta transformation formulas}, we extend the $I$-Bessel lattice sum identities to those with characters. Next we prove that some transformation formulas of theta functions of lattices can be given by these extended $I$-Bessel lattice sum identities. In particular, we treat an example on the Dedekind eta function (see Example \ref{ex:eta}).   
In Section \ref{Lattice sums of I-Bessel functions coming from linear codes}, we review terminologies on linear codes over rings. We consider the lattice sums coming from linear codes and prove the identities of them involving complete weight enumerators of linear codes.
In Section \ref{Sums of I-Bessel functions and heat equations on lattices}, we consider the heat equations on lattices and 
obtain their solutions as $I$-Bessel lattice sums.  
The $I$-Bessel lattice sum corresponding to the Dedekind eta function is interpreted as the solution of the heat equation on a certain lattice (see Remark \ref{heat eq and eta}).
Finally, we obtain the solution of the heat equation on $\Z^n$ 
related to a linear code.

\medskip
\noindent
{\bf Notation.}
The symbol $\ii$ denotes a square root of $-1$ in $\C$.
For a non-zero complex number $t=re^{\ii\theta}$ with $r>0$ and $\theta\in(-\pi,\pi]$, we set $\arg t :=\theta$ and define a square root of $t$ by $\sqrt{t}=\exp(\frac{1}{2}(\log|t|+\ii\arg t))$.
Set $\sqrt{0}:=0$.
  
Let $\ZZ_{\ge 0}$ and $\ZZ_{\ge 1}$ be the set of non-negative integers and that of positive integers, respectively.
The symbols $\R_{\ge 0}$ and $\R_{>0}$ denote the set of all non-negative real numbers and that of all positive real numbers,
respectively.

The symbol $\ll$ is Vinogradov's notation. Suffixes attached to $\ll$ mean that the implied constant there depends on those suffixes.
For complex valued functions $f$ and $g$ on a set, we write $f\asymp g$ if both $f\ll g$ and $g\ll f$ hold.

\section{Formulas of lattice sums of $I$-Bessel functions}
\label{Formulas of lattice sums of I-Bessel functions}

Let $(\R^n,\langle \cdot , \cdot \rangle)$ be the $n$-dimensional Euclidean space with the scalar product 
$\langle x, y\rangle = x\,{}^ty=\sum_{j=1}^n x_jy_j$ for $x=(x_1,\dots,x_n),
y=(y_1,\dots,y_n)\in\R^n$. The norm of $x\in \R^n$ is denoted by $\|x\|:=\sqrt{\langle x, x \rangle}$.
Throughout this article, we always treat row vectors instead of column vectors
by following the convention of coding theory.
A lattice $\Gamma$ in $\R^n$ is a $\Z$-submodule such that there is an $\R$-basis $\{a_1,\dots,a_n\}$ of $\R^n$ 
which generates $\Gamma$ as a $\Z$-module.
Let $A$ be the $n\times n$ matrix such that its $j$th row vector is $a_j$ ($j=1,\dots, n$).
Note that $\Gamma=\Z^nA$.
The Gram matrix $G$ for $A$ is defined as $G=A\,\t A =[\langle a_i , a_j \rangle]_{1\le i,j \le n}$.
There exists a unique Haar measure $dy$ of $\R^n/\Gamma$ such that
$$\int_{\R^n}f(x)dx=\int_{\R^n/\Gamma}\sum_{\gamma \in \Gamma}f(y+\gamma) dy, \qquad f \in L^1(\R^n).$$
Then we note $$\vol(\R^n/\Gamma)=|\det A|=\sqrt{\det G}.$$
For a lattice $\Gamma=\Z^nA$ in $\R^n$, $\Gamma^\ast$ is the dual lattice of $\Gamma$ defined by
$$
\Gamma^\ast:=\{ \gamma^\ast \in \R^n \mid \mbox{$\langle \gamma ,\gamma^\ast \rangle \in \Z$ for all $\gamma\in\Gamma$}\}.
$$
Then, $\Gamma^\ast=\Z^n\,\t A^{-1}$
and $(\Gamma^*)^*=\Gamma$ hold.
Moreover, we have $\vol(\R^n/\Gamma^*)=\frac{1}{\vol(\R^n/\Gamma)}$.

In the usual Poisson summation formula, test functions are imposed to be continuous.
In this article, we drop the continuity of test functions 
as follows.

\begin{theorem}{\rm (The Poisson summation formula)}
\label{PSF}
Let $\Gamma$ be a lattice in $\R^n$. 
Let $f:\R^n \to \C$ be an $L^1$-function.
Then the Fourier transform 
$$
\hat{f}(\xi):=\int_{\R^n}f(x)e^{-2\pi \ii\langle x, \xi \rangle}\,dx, \qquad \xi \in \R^n
$$
of $f$ is defined. Suppose the following.
\begin{enumerate}
\item[(i)] The series $\sum_{\gamma \in\Gamma} f(x+\gamma)$ converges absolutely
and continuous as a function in $x \in \R^n$.

\item[(ii)] The series 
$
\sum_{\gamma^\ast \in \Gamma^\ast} \hat{f}(\gamma^\ast)
$
is absolutely convergent.
\end{enumerate}
Then, we have the identity
\begin{align}\label{Fourier exp of series}
\sum_{\gamma \in \Gamma} f(x+\gamma)
=\dfrac{1}{\vol(\R^n/\Gamma)}\sum_{\gamma^\ast \in \Gamma^\ast}
\hat{f}(\gamma^\ast)e^{2\pi \ii \langle x, \gamma^\ast \rangle}, \qquad x \in \R^n.
\end{align}
\end{theorem}
\begin{proof}
The series $\sum_{\gamma \in\Gamma} f(x+\gamma)$ is regarded as a function on the compact abelian group $\R^n/\Gamma$.
Since the Pontrjagin dual of $\Gamma$ is equal to $\{e^{2\pi \ii\langle x, \gamma^*\rangle}\mid \gamma^*\in \Gamma^*\}$ and
the set $\{\vol(\R^n/\Gamma)^{-1/2}e^{2\pi \ii\langle x, \gamma^*\rangle}\mid \gamma^*\in \Gamma^*\}$
is a complete orthonormal system of $L^2(\R^n/\Gamma)$,
we obtain the equality \eqref{Fourier exp of series}
as an element of $L^2(\R^n/\Gamma)$.
As the both sides of \eqref{Fourier exp of series} are continuous by (i) and (ii), the equality \eqref{Fourier exp of series} holds for all $x \in \R^n$.
\end{proof}

\begin{remark}\label{conti of series}
	In our Poisson summation formula, our test functions $f$ are not necessarily continuous.
However, we need the continuity of the series $\sum_{\gamma \in\Gamma} f(x+\gamma)$.
For example, if $f$ is the characteristic function of $[-1/2, 1/2]$ on $\R$, then $f$ is not continuous.
In that case, $\sum_{\gamma \in \ZZ}f(x+\gamma)$ is not continuous at the points in  $1/2+\Z$.
For the use of the Poisson summation formula, we need to replace this $f$ with the rectangular function
\begin{align}\label{ch}
\1_{[-1/2,1/2]}(x)
=
\begin{cases}
1 &  \text{$($if $|x| <1/2)$,}\\
1/2 & \text{$($if $|x| = 1/2)$,}\\
0 &  \text{$($if $|x|>1/2)$.}
\end{cases}
\end{align}
Then $\sum_{\gamma \in \ZZ}\1_{[-1/2,1/2]}(x+\gamma)$ is continuous on $\R$. This idea is used in the proofs of Theorems \ref{thm:I-E}
and \ref{thm:c}.
\end{remark}

We consider the $I$-Bessel function (the modified Bessel function of the first kind) defined as
$$I_x(t):=\left(\frac{t}{2}\right)^{x}\sum_{m=0}^{\infty}\frac{1}{m!\Gamma(m+x+1)}\left(\frac{t}{2}\right)^{2m}, \qquad t \in \C-(-\infty, 0]
$$
for a fixed $x\in \C$, where $(t/2)^x$ is defined as $\exp(x\log (t/2))$ with $\arg \log (t/2) \in (-\pi,\pi)$.
This series is a solution of the modified Bessel differential equation
$$t^2\frac{d^2w}{dt^2}+t\frac{dw}{dt}-(x^2+t^2)w=0.$$
If $x\in \Z$, the defining series makes sense for all $t \in \C$. Hence $I_x(t)$ is holomorphic on $\C$ as a function in $t$.
We have $I_{x}(t)=I_{-x}(t)$ for all $x\in \Z$ and $t \in \C$.

For $
(x,t)\in \C \times \{t\in\C \mid |\arg t|<\pi/2\} \cup \{x \in \C \mid \Re(x)>0\}\times \{t\in \C \mid  |\arg t|=\pi/2 \},
$
the $I$-Bessel function $I_{x}(t)$ has the following integral representation
\cite[p.181 (4)]{Watson}:
$$
I_{x}(t)=\dfrac{1}{\pi}\int_0^\pi e^{t\cos \theta}\cos x\theta\,d\theta
-\dfrac{\sin \pi x}{\pi}\int_0^\infty e^{-t \cosh y}e^{-xy}\,dy.
$$
For $(x,t)\in\C^2$, we set
$$
\tilde{I}_x(t) := \dfrac{1}{\pi}\int_0^\pi e^{t\cos \theta}\cos x\theta\,d\theta. 
$$
Note that $\tilde{I}_x(t)=I_{x}(t)$ holds if $x\in \Z$
and that $\tilde{I}_{-x}(t)=\tilde{I}_x(t)$ holds for all $x\in \C$.

The following lemmas are used in the proof of
Theorem \ref{thm:I-E:char} for trivial characters,
i.e., Theorem \ref{thm:I-E}.
\begin{lemma}\label{lem:Ftrans-I}
For $(x,t)\in\R\times \C$, we have
\begin{align}\label{invFtrans-I}
\tilde{I}_x(t)
=\int_\R \1_{[-1/2,1/2]} (\xi) \exp(t\cos 2\pi \xi)
e^{-2\pi\ii \xi x}\,d\xi,
\end{align}
where $\1_{[-1/2,1/2]}(x)$ is the rectangular function defined as
\eqref{ch}.
\end{lemma}
\begin{proof}
We set
$$f(\xi)=\1_{[-1/2,1/2]}(\xi) \exp(t\cos 2\pi \xi),\qquad \xi \in \R.$$
The right-hand side of \eqref{invFtrans-I} equals $\hat f(x)$ and can be written as
\begin{align*}
	&\int_{-1/2}^{1/2} e^{t\cos 2\pi \xi}
	(\cos 2\pi\xi x-\ii \sin 2\pi\xi x)\,d\xi
	\\
	=&\int_{-1/2}^{1/2}e^{t\cos 2\pi \xi}
	\cos 2\pi\xi x\,d\xi
	=\dfrac{1}{\pi}\int_0^\pi e^{t\cos \theta}\cos x\theta\,d\theta
	=\tilde{I}_x(t).
\end{align*}
This completes the proof.
\end{proof}
\begin{lemma}\label{convergence lem}
Let $\Gamma$ be a lattice in $\R^n$. For any $t_1,\ldots, t_n \in \C$, the series
$$\sum_{\gamma=(\gamma_j)_j \in \Gamma} \prod_{j=1}^n \left\{ \tilde{I}_{x_j+\gamma_j}(t_j)-\frac{e^{-t_j}\sin \pi (x_j+\gamma_j)}{\pi (x_j+\gamma_j)} \right\}$$
converges absolutely and locally uniformly on $\{(x,t) \mid x =(x_j)_{j=1}^{n} \in \R^n,\ t=(t_j)_{j=1}^{n} \in \C^n \}$.

If $\Gamma$ is a sublattice of $\Z^n$ and $V$ is a compact set in $\C^n$, then
$$\sum_{\gamma=(\gamma_j)_j \in \Gamma} \prod_{j=1}^n {I}_{x_j+\gamma_j}(t_j)$$
converges absolutely and uniformly on $\{(x,t)\mid x \in \ZZ^n, \ t \in V \}$.
\end{lemma}
\begin{proof}
We set $F_j(a):=\tilde I_a(t_j)-\frac{e^{-t_j}\sin \pi a}{\pi a}$.
By integration by parts,
$F_j(a)$ for $a\in \R-\{0\}$ is evaluated as
{\allowdisplaybreaks\begin{align*}
&F_j(a)= \frac{1}{\pi}\int_0^{\pi}e^{t_j\cos \theta}\cos a \theta d \theta 
-\frac{e^{-t_j}\sin \pi a}{\pi a}
=  \frac{1}{\pi a}\int_0^{\pi} e^{t_j \cos\theta}(t_j\sin\theta) \sin a \theta d\theta\\
= & \left[t_je^{t_j \cos\theta}\sin\theta \frac{-1}{\pi a^2}\cos a\theta\right]_0^{\pi} +\frac{1}{\pi a^2}\int_{0}^\pi t_j(e^{t_j\cos\theta}\sin\theta)' \cos a\theta d\theta\\
=& \frac{1}{\pi a^2}\int_{0}^{\pi}e^{t_j\cos\theta}(-t^2_j\sin^2\theta+t_j\cos\theta)\cos a \theta d\theta
\end{align*}
}and
\begin{align*}
F_j(0)=\frac{1}{\pi}\int_0^{\pi}e^{t_j\cos \theta}d\theta 
-e^{-t_j}.
\end{align*}
Then, for any $\delta >0$ we have
$$|F_j(a)| \ll_\delta (1+|t_j|+|t_j|^2)e^{|t_j|}\frac{1}{(1+|a|)^2}, \qquad \mbox{$|a|\ge \delta$ or $a=0$}.$$
Set  $T:=\max_{1\le j \le n} |t_j|$ and $\delta:=\min_{1\le j \le n} \min_{\gamma \in \Gamma, x_j+\gamma_j\neq0}|x_j+\gamma_j|>0 $. We obtain the estimate
$$\sum_{\gamma=(\gamma_j)_j \in \Gamma} \prod_{j=1}^{n}|F_j(x_j+\gamma_j)| \ll_\delta (1+T+T^2)^n e^{nT}
\sum_{\gamma\in\Gamma}\prod_{j=1}^{n}\frac{1}{(1+|x_j+\gamma_j|)^2},$$
where the implied constant is independent of $t$ and $x$.
For any compact set $K \subset \R^n$,
we have the estimate $1+|x_j+\gamma_j|\asymp 1+|\gamma_j|$ uniformly in $x\in K$ and $\gamma \in \Gamma$.
By the same argument in \cite[Lemma A.5]{SugiyamaTsuzuki}, we have
\begin{align}\label{convergence of series}
\sum_{\gamma\in\Gamma}\prod_{j=1}^{n}\frac{1}{(1+|x_j+\gamma_j|)^2}\ll_{\Gamma,K} \int _{\R^n}\left\{\prod_{j=1}^n\frac{1}{(1+|y_j|)^2}\right\}dy<\infty.
\end{align}

Suppose that $\Gamma$ is a sublattice of  $\Z^n$ and $x\in \Z^n$.
By
$$|I_{a}(t_j)|=|\tilde I_{a}(t_j)|\le |F_j(a)|+{\delta_{a,0}}e^{|t_j|}\ll (1+|t_j|+|t_j|^2)e^{|t_j|}\frac{1}{(1+|a|)^2} , \ a \in \Z,$$
where $\delta_{x,y}$ is the Kronecker delta function on $\ZZ\times \ZZ$,
we have the following estimate:
\begin{align*}
	\sum_{\gamma=(\gamma_j)_j \in \Gamma} \prod_{j=1}^n 
	|I_{x_j+\gamma_j}(t_j)|
	\ll & 
	 (1+T+T^2)^ne^{nT}
	\sum_{\gamma=(\gamma_j)_j \in \Gamma}	\prod_{j=1}^n\frac{1}{(1+|x_j+\gamma_j|)^2} \\
	\le & (1+T+T^2)^ne^{nT}
	\sum_{\gamma=(\gamma_j)_j \in \Z^n}	\prod_{j=1}^n\frac{1}{(1+|\gamma_j|)^2} <\infty.
\end{align*}
This completes the proof.
\end{proof}

\begin{remark}\label{rem not abs convergent}
The $I$-Bessel lattice sum $\sum_{\gamma \in \Gamma}\prod_{j=1}^{n}{\tilde I}_{x_{j}+\gamma_j}(t_j)$ 
is not absolutely convergent for a general lattice $\Gamma$ and a general $x \in \R^n$
by the first assertion of Lemma \ref{convergence lem}.
\end{remark}

The following is the main theorem in this section. 
\begin{theorem}\label{thm:I-E}
	Let $\Gamma$ be a sublattice of $\Z^n$.
	Take $t_1,\dots,t_n \in \C$. 
	If $x\in \Z^n$ and $y\in \R^n$ then we have 
	\begin{align}\label{I-E-1}
		&\sum_{\gamma \in \Gamma} \left\{\prod_{j=1}^n I_{x_j+\gamma_j}(t_j )\right\}e^{2\pi \ii\langle y, x+\gamma\rangle}\\
		=&
		\dfrac{1}{\vol(\R^n/\Gamma)}
		\sum_{\gamma^\ast \in \Gamma^\ast}
		\left\{\prod_{j=1}^{n}\1_{[-1/2,1/2]}(y_j-\gamma_j^*)\exp(t_j\cos 2\pi (y_j-\gamma_j^\ast))\right\}e^{2\pi \ii\langle x, \gamma^\ast \rangle}, \notag
	\end{align}
where ${\bf 1}_{[-1/2,1/2]}$ is the rectangular function given as \eqref{ch}.
		In particular, if  $y=0$ then we have 
	\begin{align}\label{I-E}
		&\sum_{\gamma \in \Gamma} \prod_{j=1}^n I_{x_j+\gamma_j}(t_j )\\
		=&
		\dfrac{1}{\vol(\R^n/\Gamma)}
		\sum_{\gamma^\ast \in \Gamma^\ast}
		\left\{\prod_{j=1}^{n}\1_{[-1/2,1/2]}(\gamma_j^*)\exp(t_j\cos 2\pi \gamma_j^\ast)\right\}e^{2\pi \ii\langle x,\gamma^\ast \rangle}. \notag
	\end{align}
	Here,
$\gamma_j$ and $\gamma_j^\ast$ are the $j$th components of $\gamma$ and of $\gamma^\ast$, respectively.
\end{theorem}
\begin{proof}
	The equation \eqref{I-E} follows immediately from \eqref{I-E-1}.
	In the following, we prove \eqref{I-E-1} by replacing $\Gamma$ with $\Gamma^*$.
	We suppose that $\Gamma^*$ is a sublattice of $\Z^n$.
	Set $f(z):=\prod_{j=1}^{n}\1_{[-1/2,1/2]} (z_j)e^{t_j\cos 2\pi z_j}$ for any $z=(z_1,\ldots,z_n) \in \R^n$.
Take $x \in \Z^n$ and $y \in \R^n$. Then the function
	$$
	g(z):=f(z)e^{-2\pi\ii\langle x, z \rangle}, \qquad (z\in\R^n)
	$$
is a compactly supported non-continuous function.
Now $\sum_{\gamma \in \Gamma}g(z+\gamma)$	is a finite sum, and continuous as a function in $z\in \R^n$
as we see in Remark \ref{conti of series}.
We have the estimate
	\begin{align}\label{majorant of FT}
		\sum_{\gamma^* \in \Gamma^*}
		|\hat{g}(\gamma^*)|
		&=
		\sum_{\gamma^* \in \Gamma^*}
		|{\hat f}(\gamma^*+x)|
\le \sum_{\gamma^*\in \Z^n}|\hat f(\gamma^*)|
	\end{align}
 with the aid of $\Gamma^*\subset \Z^n$ and $x \in \Z^n$
	and the majorant series converges absolutely by
	Lemmas \ref{lem:Ftrans-I} and \ref{convergence lem}.
	By applying the Poisson summation formula
	(Theorem \ref{PSF}), we have
	\begin{align}\label{phi-f=phi-f}
		\sum_{\gamma \in \Gamma}f(y+\gamma)e^{-2\pi\ii\langle x, y+\gamma \rangle}=
		\frac{1}{\vol(\R^n/\Gamma)}\sum_{\gamma^* \in \Gamma^*}
		\hat f(\gamma^*+x)e^{2\pi \ii\langle y,\gamma^*\rangle}.
	\end{align}
	As a consequence, we obtain
	{\allowdisplaybreaks\begin{align*}
	& \sum_{\gamma \in \Gamma}
\left\{	\prod_{j=1}^{n}\1_{[-1/2,1/2]}(y_j-\gamma_j)\exp(t_j\cos 2\pi (y_j-\gamma_j))\right\}e^{2\pi\ii\langle x, \gamma \rangle}
		\\
	=	& e^{2\pi\ii\langle x,y \rangle} \times \sum_{\gamma \in \Gamma}
\left\{	\prod_{j=1}^{n}\1_{[-1/2,1/2]}(y_j+\gamma_j)\exp(t_j\cos 2\pi (y_j+\gamma_j))\right\}e^{-2\pi\ii\langle x, y+\gamma \rangle}
		\\
		= & e^{2\pi\ii\langle y,x \rangle} \times \dfrac{1}{\vol(\R^n/\Gamma)}\sum_{\gamma^* \in \Gamma^*} \left\{\prod_{j=1}^n I_{x_j+\gamma^*_j}(t_j)\right\}e^{2\pi \ii\langle y, \gamma^*\rangle}.
	\end{align*}
}By the identity $\vol(\R^n/\Gamma^*)=\dfrac{1}{\vol(\R^n/\Gamma)}$ and exchanging the symbols $\Gamma$ for $\Gamma^*$, we obtain
	\eqref{I-E-1}.
\end{proof}

We now introduce the symbol $\sum^{'}$, 
which is often used in this article. 
For any $a \in \ZZ_{\ge 0}$ and a sequence $(c_k)_{k\in\Z}$ in $\C$, 
we set 
$$
\sideset{}{'}\sum_{k=-[a/2]}^{[a/2]} c_k
:=
\left\{
\begin{array}{ll}
\displaystyle\frac{1}{2} c_{-a/2}+\sum_{k=-a/2+1}^{a/2-1} c_k+\frac{1}{2}c_{a/2},
&\mbox{if $a$ is an even integer,}
\\[6pt]
\displaystyle\sum_{k=-[a/2]}^{[a/2]} c_k,
&\mbox{otherwise.}
\end{array}
\right.
$$

From now, we present several examples of Theorem \ref{thm:I-E}.

\begin{example}	Take $t\in \C$, $m \in \ZZ_{>0}$ and $x\in \Z$. 
If $\Gamma=m\Z$ then $\Gamma^\ast= \frac{1}{m}\Z$ and $\vol(\R/\Gamma)=m$ hold. 
Since $x$ is an integer, we have 
{\allowdisplaybreaks
\begin{align}\label{I-integer} 
	\sum_{\gamma \in m\Z} {I}_{x+\gamma }(t)
&=
\dfrac{1}{m}
\sum_{\gamma^\ast \in \frac{1}{m}\Z }\1_{[-1/2,1/2]}(\gamma^*)
\exp(t\cos 2\pi \gamma^\ast)e^{2\pi \ii x \gamma^\ast}
\\
&=
\dfrac{1}{m}
\sideset{}{'}\sum_{k=-[\,m/2\,]}^{[\,m/2\,]}
\exp\left(t\cos \frac{2\pi k}{m}\right)e^{2\pi \ii x \frac{k}{m}}
\notag 
\\
&=
\dfrac{1}{m}
\sum_{k=0}^{m-1}\exp\left(t\cos \frac{2\pi k}{m}\right)e^{2\pi \ii x \frac{k}{m}}.
\notag
\end{align}
}The identity \eqref{I-integer} is the same as in \cite[Theorem 2]{KN}
{\rm (}cf.\ \cite[Theorem 1]{ADG} for \eqref{I-integer}  at $x=0${\rm )}. 
\end{example}

\begin{example}\label{ex:discretetorus}
Take $x=(x_1,\ldots,x_n)\in\Z^n$ and $t_1,\dots,t_n \in \C$.  Let $m_1,\dots, m_n$ be positive integers. 
If $\Gamma=\prod_{j=1}^n m_j \Z$ then the equalities
$\Gamma^\ast=\prod_{j=1}^n \frac{1}{m_j} \Z$ and $\vol(\R^n/\Gamma)=m_1\cdots m_n$ hold. We have 
\begin{align}\label{I-E-discretetorus}
	&
	\sum_{(k_1,\dots,k_n)\in \Z^n}
	\prod_{j=1}^n I_{x_j+m_j k_j}(t_j)\\
	=&
	\dfrac{1}{m_1\cdots m_n}
	\sum_{(k_1,\dots, k_n) \in \prod_{j=1}^n (\Z \cap [0, m_j-1])} 
\prod_{j=1}^{n}	\exp\left(t_j\cos \frac{2\pi k_j}{m_j}\right)
	e^{2\pi \ii \frac{x_j k_j}{m_j}}.\notag
\end{align}
Furthermore, if $x=0$ and $t_1=\dots =t_n=2t$ $(t\in\C)$ then 
by multiplying the both sides of \eqref{I-E-discretetorus} by $e^{-2nt}$, 
we have 
\begin{align}\label{discrete-torus}
\sum_{(k_1,\dots,k_n)\in \Z^n}
\prod_{j=1}^ne^{-2t}I_{m_j k_j}(2t)
=
\dfrac{1}{\vol(\R^n/\Gamma)}\sum_{\lambda \in \Spec(\Delta_{{\rm DT}_m})}e^{-\lambda t},
\end{align}
where $$
\Spec(\Delta_{{\rm DT}_m})
=\left\{2n-2\sum_{j=1}^n \cos \dfrac{2\pi k_j}{m_j}  \ {\bigg| }\ 0\le k_j <m_j, j=1,2,\dots,n\right\}
$$
is the multiset of all eigenvalues of the Laplacian $\Delta_{{\rm DT}_m}=2nI_{m_1\cdots m_n}-A_{{\rm DT}_m}$ of the discrete torus
${\rm DT}_m:=\prod_{j=1}^n C_{m_j}$
with $m:=(m_j)_j$.
Here $A_{{\rm DT}_m}$ is the adjacency matrix of ${\rm DT}_m$ and $C_{m_j}$ is the cycle graph with $m_j$ vertices. 
The identity \eqref{discrete-torus} is the same as in \cite[Lemma 3.1]{CJK}.  
\end{example}
\begin{lemma}\label{lem:eigenvalue}
Let $\Gamma=\Z^n A$ be a lattice of $\R^n$. 
Let $\{ e_1 = (1,0,\dots,0), \dots, e_n = (0,\dots,0,1) \}$ 
be the standard basis of $\R^n$.
For a bounded function $f : \R^n/\Gamma  \to \C$, 
let $\Delta'$ denote the Laplacian defined by 
\begin{align}\label{laplacian'}
\Delta' f(x) := \dfrac{1}{2}\sum_{i=1}^n  \left( f(x+e_i)+f(x-e_i)-2f(x) \right)
\end{align}
{\rm (}cf.\ \cite[pp.74--75]{Borthwick}{\rm )}. For a fixed $\gamma^\ast=(\gamma_1^\ast,\ldots, \gamma_n^\ast) \in \Gamma^\ast$, 
we define $g_{\gamma^\ast} : \R^n/\Gamma \to \C$ as
$$
g_{\gamma^\ast} (x)
:=e^{2\pi \ii\langle x,\gamma^\ast \rangle}.
$$
Then the function $g_{\gamma^\ast}(x)$ 
is an eigenfunction of $\Delta'$ of the eigenvalue
$$\sum_{i=1}^n(\cos 2\pi  \gamma^\ast_i -1).$$
Furthermore, the multiset of all eigenvalues of $\Delta'$ 
is 
\begin{align}\label{eigen of Delta'}
\left\{ \sum_{i=1}^n(\cos 2\pi   \gamma^\ast_i -1) \mid \gamma^\ast=(\gamma_1^\ast,\ldots, \gamma_n^\ast) \in \Gamma^\ast \right\}.
\end{align}
\end{lemma}
\begin{proof}
By definition we see 
\begin{align*}
\Delta' g_{\gamma^\ast}(x)
&=\dfrac{1}{2}\sum_{i=1}^n  \left( g_{\gamma^\ast}(x+e_i)+g_{\gamma^\ast}(x-e_i)-2g_{\gamma^\ast}(x) \right)
\\
&=
\dfrac{1}{2}
\sum_{i=1}^n
\left(
e^{2\pi \ii\langle x+e_i,\gamma^\ast \rangle}+e^{2\pi \ii\langle x-e_i,\gamma^\ast \rangle}
-2e^{2\pi \ii\langle x,\gamma^\ast \rangle}
\right). 
\end{align*}
By the relation
$$
e^{2\pi \ii\langle x+e_i,\gamma^\ast \rangle}+e^{2\pi \ii\langle x-e_i,\gamma^\ast \rangle}
-2e^{2\pi \ii\langle x,\gamma^\ast \rangle}
=2e^{2\pi \ii\langle x,\gamma^\ast \rangle}
\left(
\cos 2\pi  \langle e_i ,\gamma^\ast \rangle -1
\right),
$$
we obtain $\Delta' g_{\gamma^\ast}(x)=\sum_{i=1}^n(\cos 2\pi \gamma_i^\ast -1) g_{\gamma^\ast}(x)$. 
The set 
$\{\tfrac{1}{\sqrt{\vol(\R^n/\Gamma)}}g_{\gamma^\ast}(x)\}_{\gamma^\ast \in \Gamma^\ast}$ is a complete orthonormal system of $L^2(\R^n/\Gamma)$.  
Thus the multiset of all eigenvalues of $\Delta'$ 
is equal to
$\{ \sum_{i=1}^n(\cos 2\pi   \gamma^\ast_i -1) \mid \gamma^\ast \in \Gamma^\ast \}.$
\end{proof}
\begin{remark}\label{rem:partial trace}
Assume that $\Gamma=\Z^nA$ is a sublattice of $\ZZ^n$.
By Lemma \ref{lem:eigenvalue}, the right-hand side of  \eqref{I-E} is a sum of eigenfunctions of $\Delta'$ 
over the subset
$$
\{ \sum_{i=1}^n(\cos 2\pi   \gamma^\ast_i -1) \mid \gamma^\ast=(\gamma_1^\ast,\ldots, \gamma_n^\ast) \in \Gamma^\ast  \cap [-1/2,1/2]^n \}
$$
of the multiset \eqref{eigen of Delta'}.
By setting $t_1=\cdots=t_n=t/n$ $(t\in\C)$ and multiplying the both sides of \eqref{I-E} by $e^{-t}$,  we obtain  
\begin{align*}
	& \sum_{\gamma \in \Gamma} e^{-t}\prod_{j=1}^n I_{x_j+\gamma_j}(\tfrac{t}{n})\\
	= &
	\dfrac{1}{\vol(\R^n/\Gamma)}
	\sum_{\gamma^\ast \in \Gamma^\ast}
	\left\{\prod_{j=1}^{n}\1_{[-1/2,1/2]}(\gamma_j^*)\exp(\tfrac{t}{n}(\cos 2\pi \gamma_j^\ast-1))\right\}e^{2\pi \ii\langle x,\gamma^\ast \rangle}. \notag
\end{align*}
Assume $\Gamma=\Z^n$.
By the same argument as in \cite[\S3]{KN} {\rm (}cf.\ \cite[\S 5.2]{Slavik}{\rm )},
the function $\K_{\Z^n,t}(x):=e^{-t}\prod_{j=1}^n I_{x_j}(t/n)$, $(x=(x_1,\ldots,x_n) \in \Z^n)$ in the left-hand side
above is a solution to the semidiscrete heat equation on $\Z^n$ defined by
\begin{align*}
\begin{cases}
&		\partial_t \, \K_{\Z^n,t}(x) = \frac{1}{n}\Delta' \, \K_{\Z^n,t}(x), \qquad x \in \Z^n, \\
& \displaystyle \lim_{t\rightarrow +0}\K_{\Z^n,t}(x) = \delta_{x,0}, \qquad x \in \Z^n,
\end{cases}
\end{align*}
where $\delta_{x,y}$ is the Kronecker delta function on $\Z^n\times\Z^n$.
Note that the Laplacian $\Delta$ as \eqref{laplacian} in Section \ref{Sums of I-Bessel functions and heat equations on lattices} is 
different from $\Delta'$ in \eqref{laplacian'}
and that we have the relation $\Delta'=n\Delta_{\Z^n}$.
\end{remark}

\section{Generalizations of sums of the $I$-Bessel functions on lattices with characters and theta transformation formulas}
\label{Generalizations of sums of the I-Bessel functions on lattices with characters and theta transformation formulas}
In this section, we consider results for the theta functions 
of sublattices of $\Z^n$ and extend them with characters. 

Let $\Gamma = \Z^n A$ be  
a lattice in $\R^n$ with $A$ being an $n$-by-$n$ regular matrix.
We call
$$
\Theta_\Gamma (t) 
:=
\sum_{\gamma \in \Gamma} e^{-(2\pi)^2 \langle \gamma, \gamma \rangle t} 
=
\sum_{x \in \Z^n} e^{-(2\pi)^2 \langle xA, xA \rangle t} 
= 
\sum_{x \in \Z^n} e^{-(2\pi)^2 \langle xG, x \rangle t} ,\quad t>0
$$
the theta function associated to the real torus $\R^n/\Gamma$, 
where $G= A \, \t A$ is the Gram matrix of $A$. 
This theta function is often called the theta function of a lattice $\Gamma$  
(e.g., \cite{Conway-Sloane}, \cite[Section 9.1.1]{NRS} and \cite[Chapter 2]{Ebeling}).
We note that the symbol $\Theta_A(t)$ in \cite{CJK} is equal to $\Theta_{\Gamma^*}(t)$
in our notation\footnote{The dual lattice $A^*$ to $A$ in \cite[\S2.5]{CJK} is not difined there. In \cite[\S2.5]{CJK}, $A^*$ is the Gram matrix for ${}^{t}A^{-1}$.}.

Let $\alpha_1,\ldots, \alpha_n$ be positive real numbers. 
When $A=\diag(\alpha_1,\ldots, \alpha_n)$,
Chinta, Jorgenson and Karlsson proved in \cite[Proposition 5.2]{CJK}
that the formula \eqref{discrete-torus} is transformed into
the theta inversion formula 
of $\Theta_\Gamma(t)$ in the following way:
Let $L=L(u)$ be a sequence of positive integers indexed by 
$u \in \Z_{\ge 1}$ such that $\frac{L(u)}{u} \to \alpha$ as $u \to \infty$ 
holds for some positive real number $\alpha$. 
For a fixed $t>0$ and $k \in \Z_{\ge 0}$, 
we have the pointwise limit 
\begin{align}\label{limit-I}
L e^{-2u^2 t} I_{k L} (2u^2 t) \to 
\dfrac{\alpha}{\sqrt{4 \pi t}} e^{-\frac{(\alpha k)^2}{4t}} 
\end{align}
as $u \to \infty$ by \cite[Proposition 4.7]{CJK}.
More precisely, let $N(u)=(N_1(u),\dots, N_n(u))$ be 
$n$-tuples of integer sequences parametrized by $u\in\Z_{\ge 1}$ such that 
\begin{align*}
\displaystyle\frac{1}{u}N(u)
=\left(\frac{N_1(u)}{u},\dots, \frac{N_n(u)}{u}\right)\to 
(\alpha_1,\dots,\alpha_n) 
\end{align*}
as $u \to \infty$ holds for some positive real numbers $\alpha_1,\dots,\alpha_n$. 
Substituting $u^2 t$ for $t$ and $(m_1,\dots,m_n)=(N_1(u),\dots,N_n(u))$ in  \eqref{discrete-torus}, and multiplying \eqref{discrete-torus} by 
the number $\prod_{j=1}^n N_j(u)$ of vertices of the graph ${\rm DT}_{N(u)}= \prod_{j=1}^n C_{N_j(u)}$, we obtain 
\begin{align}\label{discrete-torus-N}
& \sum_{(k_1,\dots,k_n)\in \Z^n}
\prod_{j=1}^n  N_j(u)\, e^{-2u^2t}I_{N_j(u) k_j}(2u^2t)\\
= &
\dfrac{\prod_{j=1}^n N_j(u)}{\vol(\R^n/\Gamma(N(u)))}\sum_{\lambda \in \Spec(\Delta_{{\rm DT}_{N(u)}})}e^{-\lambda u^2 t}, \notag
\end{align}
where $\Gamma(N(u)):=\Z^n\diag(N_1(u),\dots, N_n(u))$.

As $u\to\infty$, the left-hand side and the right-hand side of \eqref{discrete-torus-N} converge to  the left-hand side and the right-hand side of the following 
identity, respectively:
$$
|\det A|(\tfrac{1}{4\pi t})^{n/2} \Theta_{\Gamma}(\tfrac{1}{16\pi^2 t})=\Theta_{\Gamma^*} (t), 
$$
where $\Gamma:=\Z^n\diag(\alpha_1,\dots,\alpha_n)$.
Now, we review the absolute convergence of the sums of the Bessel functions on lattices appearing in \eqref{discrete-torus-N}. 
\begin{lemma}\label{lem:L-sum}
Let $t_1,\dots,t_n$ be any complex numbers and let $m$ be a positive integer.
Let $\Gamma$ be a sublattice of $\Z^n$. 
Take $(x_1,\ldots, x_n) \in \Z^n$. 
Then the sum
$$
\sum_{(\gamma_1,\dots,\gamma_n)\in \Gamma } \prod_{j=1}^n \sqrt{t_j}\cdot Le^{-L^2 |t_j|}   I_{(\gamma_j+x_j) L}(L^2 t_j)
$$
is bounded for $L\in\Z_{\ge 1}$.

In particular, for $b\in \ZZ$ the sum  
$$
\sqrt{t}\cdot Le^{-L^2 |t|} \sum_{k\in \Z}  I_{(mk+b)L}(L^2 t)
$$
is bounded for $L\in\Z_{\ge1}$.

\end{lemma}
\begin{proof}
We may assume $t_j\neq0$ for all $j\in \{1,\ldots,n\}$.
For $z\in\C$ with $|\arg z|<\pi$ and $\nu \in \R$, the identity 
$
I_\nu(z)=\sum_{m\ge 0} \frac{(z/2)^{\nu+2m}}{m!\Gamma(m+\nu+1)}
$	
implies
$$
|I_x(t)| =|I_{|x|}(t)| \le \sum_{m\ge 0} \dfrac{(|t|/2)^{|x|+2m}}{m!\Gamma(m+|x|+1)}=I_{|x|}(|t|)
$$
for $t\in\C$ and $x\in\Z$.
By the inequality
$$
\sqrt{s}\cdot e^{-s} I_{|x|}(s) \le \left(1+\frac{|x|}{s}\right)^{-\frac{|x|}{2}}
$$
for $s> 0$ and $x \in \Z$ in \cite[Corollary 4.4]{CJK},
we have
$$
|\sqrt{t}\cdot e^{-|t|} I_{x}(t)| \le 
\sqrt{|t|}\cdot e^{-|t|} I_{|x|}(|t|)\le \left(1+\frac{|x|}{|t|}\right)^{-\frac{|x|}{2}}
$$
for $t\in\C-\{0\}$ and $x\in\Z$. 
Under the condition $\gamma_j ,x_j , L\in \Z$ ($j=1,\dots,n$),
we obtain
\begin{align*}
&
\left|\sum_{(\gamma_1,\dots,\gamma_n)\in \Gamma } \prod_{j=1}^n \sqrt{t_j}\cdot Le^{-L^2 |t_j|}   I_{(\gamma_j+x_j) L}(L^2 t_j)\right|
\\
\le &
\sum_{(\gamma_1,\dots,\gamma_n)\in \Gamma} \prod_{j=1}^n
\left(
1+\dfrac{|\gamma_j+x_j|}{L|t_j|}
\right)^{-\frac{|\gamma_j+x_j|L}{2}}.
\end{align*}
By the monotone convergence theorem,
the limit of the right-hand side above as $L\rightarrow \infty$ is estimated as
\begin{align*}& \lim_{L\rightarrow \infty}
\sum_{\gamma=(\gamma_1,\dots,\gamma_n)\in \Gamma} \prod_{j=1}^n
\left(
1+\dfrac{|\gamma_j+x_j|}{L|t_j|}
\right)^{-\frac{|\gamma_j+x_j|L}{2}} \\
= &
\sum_{\gamma\in \Gamma}
\lim_{L\rightarrow \infty}\prod_{j=1}^n
\left(
1+\dfrac{|\gamma_j+x_j|}{L|t_j|}
\right)^{-\frac{|\gamma_j+x_j|L}{2}} \\
= &
\sum_{\gamma\in \Gamma} e^{-\sum_{j=1}^{n}\frac{|x_j+\gamma_j|^2}{2|t_j|}}
\le
\sum_{\gamma\in \Gamma} e^{-\frac{\langle \gamma+x, \gamma+x  \rangle}{2T}}
\ll \sum_{\gamma \in \Gamma}\prod_{j=1}^{n}\frac{1}{(1+\frac{1}{T^{1/2}}|x_j+\gamma_j|)^2}
\\ 
= & \sum_{\gamma \in \Gamma}\prod_{j=1}^{n}\frac{T}{(T^{1/2}+|x_j+\gamma_j|)^2}
\le \sum_{\gamma \in \Gamma}\prod_{j=1}^{n}\frac{T}{(1+|x_j+\gamma_j|)^2}
<\infty,
\end{align*}
where we set $T:=\max_{j=1,\ldots,n}|t_j|+1>1$ and use \eqref{convergence of series}.
Consequently, the sum in the assertion is bounded for $L \in \Z_{\ge 1}$.
\end{proof}
We can prove a kind of twisted Poisson summation formula for non-continuous test functions.
\begin{theorem}[The twisted Poisson summation formula]
\label{thm:c}
Let $q$ be a positive integer 
and let $c: \Z^n \to \C$ be a function such that 
$$
c(a_1+q,\dots,a_n+q)=c(a_1,\dots,a_n)
$$
for all $a=(a_1,\dots,a_n)\in\Z^n$. 
We regard $c$ as a function on $(\Z/q\Z)^n$.
Let 
$$
\hat{c}(m) := \sum_{a\in (\Z/q\Z)^n } c(a) e^{2\pi \ii\langle \frac{a}{q}, m \rangle} \quad (m \in \Z^n) 
$$
be the discrete Fourier transform of $c$. 
Let $\Gamma = \Z^n A$ be a lattice in $\R^n$.
Let $f:\R^n \to \C$ be an $L^1$-function 
satisfying the assumptions
(i) and (ii) in Theorem \ref{PSF} for $\Gamma$.
Then we have 
$$
\sum_{\gamma\in\Gamma} c(\gamma A^{-1} )f(x+\tfrac{\gamma}{q})
=\dfrac{1}{|\det A|}
\sum_{\gamma^\ast \in \Gamma^\ast}
\hat{c} (\gamma^\ast \t A) \hat{f} (\gamma^\ast) e^{2\pi \ii\langle x, \gamma^\ast \rangle}, \qquad x \in \R^n.
$$
\end{theorem}
\begin{proof}
We prove the assertion by generalizing the argument of Elkies' lecture note \cite{Elkies} for $n=1$.
First, we consider the case where $\Gamma=\Z^n$.  
For $x \in \R^n$, we set 
$F(x):=\sum_{k \in \Z^n} f(x+k)$. 
From Theorem \ref{PSF}, 
we have $F(x)=\sum_{m \in \Z^n} \hat{f}(m) e^{2 \pi \ii\langle x, m \rangle}$,
where we use $(\Z^{n})^* = \Z^{n}$. 
This implies
$$
\sum_{k \in \Z^n} c(k) f(x + \tfrac{k}{q}) 
= \sum_{a \in (\Z/q\Z)^n} c(a) F(x + \tfrac{a}{q})
= \sum_{m \in \Z^n} \hat{c}(m) \hat{f}(m) e^{2\pi \ii\langle x, m \rangle}. 
$$
In particular we have 
$\sum_{k\in\Z^n} c(k) f(\tfrac{k}{q})= \sum_{m \in \Z^n} \hat{c}(m) \hat{f}(m)$.

Next we consider the case where $\Gamma=\Z^n A$. 
For a fixed $x \in \R^n$, we set 
$f_{x,A}(y) := f(x+yA)$, ($y \in \R^n$). 
The function $f_{x,A}$ is an $L^1$-function satisfying the assumption (i) in Theorem \ref{PSF} for $\Z^n$.
The Fourier transform of $\hat{f}_{x,A}$ is computed as
\begin{align*}
& \hat{f}_{x,A}(m)
= \int_{\R^n} f(x+ yA) e^{-2 \pi \ii\langle y, m \rangle} \, dy 
= \dfrac{1}{|\det A|} \int_{\R^n} f(z) e^{-2 \pi \ii\langle (z-x) A^{-1}, m \rangle} \, dz 
\\
&= \dfrac{1}{|\det A|} \int_{\R^n} f(z) e^{-2 \pi \ii\langle z-x, m \t A^{-1} \rangle} \, dz 
= \dfrac{1}{|\det A|} \hat{f}(m \t A^{-1}) e^{2 \pi \ii\langle x, m \t A^{-1} \rangle}. 
\end{align*}
By this formula, $f_{x,A}$ satisfies (ii) in Theorem \ref{PSF} for $\Z^n$.
Taking $f_{x,A}$ instead of $f$ in the identity of the first case, 
we have $\sum_{k\in\Z^n} c(k) f_{x,A}(\tfrac{k}{q})
=\sum_{m\in\Z^n}\hat{c}(m)\hat{f}_{x,A}(m)$. 
Hence, we obtain 
$$
\sum_{k \in \Z^n} c(k) f(x+ \tfrac{1}{q} kA) 
= \dfrac{1}{|\det A|} \sum_{m \in \Z^n} \hat{c}(m) 
\hat{f}(m \t A^{-1}) e^{2\pi \ii\langle x,  m (\t A)^{-1} \rangle}.
$$
This completes the proof.
\end{proof}

Take an integer $q$ such that $q\ge2$. 
Let $\chi:(\Z/q\Z)^\times\to\C^\times$ be a Dirichlet character modulo $q$. 
The function $\chi$ is extended to a function $\tilde{\chi}$ on $\Z$ by 
$$
\tilde{\chi}(a)=\left\{
\begin{array}{ll}
	\chi(a \ \mod\, q) & \textrm{(if $a\in \ZZ$ and $\gcd(a,q)= 1$),}\\
	 0 & \textrm{(if $a\in \ZZ$ and $\gcd(a,q)\neq 1$).}
	\end{array}
\right.
$$
By abuse of notation, we write $\chi$ for $\tilde{\chi}$.  
Let $\1_q$ be the principal Dirichlet character modulo $q$. 
If $q=1$ then a Dirichlet character $\chi$ modulo $q$ means $\chi(a)=1$ for any $a\in\Z$. 

For a Dirichlet character $\psi$ modulo $q$, define the Gauss sum of $\psi$ by
$$
\mathcal{G}(\psi):= \sum_{a=0}^{q-1} \psi(a) e^{2\pi \ii \tfrac{a}{q}}.
$$
For $q=1$, set $\mathcal{G}(\psi):=1$.
It is well known that
$\mathcal{G}(\bar \psi)=\psi(-1)\overline{\mathcal{G}(\psi)}$.
If $\psi$ is primitive, then we have $|\mathcal{G}(\psi)|^2=\mathcal{G}(\psi)\overline{\mathcal{G}(\psi)}=\sqrt{q}$ (see \cite{Elkies}).
\begin{lemma}\label{lem:gauss-sum}
For $j=1,\dots,n$, let $\chi_j$ be a Dirichlet character modulo $q$.
Let $\chi:\Z^n\to\C$ be a function defined by 
\begin{align*}
\chi(a):=\prod_{j=1}^n \chi_j(a_j)
\end{align*}
for $a=(a_1,\dots, a_n)\in \Z^n$. 
Then 
\begin{align}\label{Gauss}
\hat{\chi}(m)= \left\{\prod_{j=1}^n\mathcal{G}(\chi_j)\right\} \bar{\chi}(m)
\end{align}
holds for all $m\in ((\Z/q\Z)^\times)^n$, where $\bar{\chi}$ is the complex conjugate of $\chi$. 
If all $\chi_j$ are primitive, then \eqref{Gauss} holds for all $m\in\Z^n$. 
\end{lemma}
\begin{proof}
For $n=1$, the identity \eqref{Gauss} is well-known (see \cite{Elkies}).
By the definition of $\chi$, we obtain 
$$
\hat{\chi}(m)
=\prod_{j=1}^n \hat{\chi_j}(m_j)
=\prod_{j=1}^n\mathcal{G}(\chi_j) \bar{\chi}(m_j)
=\left\{\prod_{j=1}^n\mathcal{G}(\chi_j)\right\} \bar{\chi}(m).
$$
This completes the proof.
\end{proof}
As a variant \eqref{chi-I-E-infty}  of the transformation formula of the Riemann theta function with a character (\cite[Chapter II]{Mumford}), we obtain Theorem \ref{thm:I-E:char} in the following way.
\begin{proof}[Proof of Theorem \ref{thm:I-E:char}]
We prove the assertion by replacing $\Gamma$ with $\Gamma^*$ as in Theorem  \ref{thm:I-E}.
Assume that $\Gamma=\Z^n\, \t A^{-1}$ is a lattice such that $\Gamma^*=\Z^n A\subset \ZZ^n$ and that every entry of $A$ is divisible by $q$.
Fix any $x \in \ZZ^n$, $y \in \R^n$ and a family $\chi=(\chi_1,\ldots,\chi_n)$ of Dirichlet characters modulo $q$.
We set $f(z):=\prod_{j=1}^n\1_{[-1/2,1/2]}(z_j)e^{t_j\cos 2\pi z_j}$.
By applying Theorem \ref{thm:c} to the function $f(z)e^{-2\pi \ii\langle x, z\rangle}$, the character $\bar\chi$
and the lattice $q\Gamma$,
we can prove the identity
$$
\sum_{\gamma\in q\Gamma} \bar\chi(\tfrac{\gamma}{q} \t A ){f}(y+\tfrac{\gamma}{q})e^{-2\pi \ii\langle x,y+\frac{\gamma}{q} \rangle}
=\dfrac{1}{|\det q\,\t A^{-1}|}
\sum_{\gamma^\ast \in (q\Gamma)^\ast}
\hat{\bar\chi} (\gamma^\ast q A^{-1}) \hat f(\gamma^\ast+x) e^{2\pi \ii\langle y, \gamma^\ast \rangle}
$$
(cf.\ \eqref{phi-f=phi-f} when $\chi$ is trivial). By $q\Gamma = \Z^n q\,\t A^{-1}$, $(q\Gamma)^\ast = \Z^n \frac{1}{q}A$ and Lemma \ref{lem:gauss-sum}, this can be written as
$$
\sum_{m\in\Z^n} \bar\chi(-m)f(y-m\,\t A^{-1})e^{-2\pi \ii\langle x, y-m\,\t A^{-1}\rangle}
=\dfrac{\prod_{j=1}^n\mathcal{G}(\bar \chi_j)}{q^n|\det\t A^{-1}|}
\sum_{k \in \Z^n}
\chi(k) \hat{f}(x+\tfrac{k A}{q}) e^{2\pi \ii\langle y,  \frac{k A}{q} \rangle}.
$$
Since every entry of $A$ is divisible by $q$,
every entry of $\frac{k A}{q}$ is contained in $\Z$.
Furthermore, the series in the right-hand side is bounded by $\sum_{\gamma \in \ZZ^n}|\hat{f}(\gamma)|$ which is convergent by Lemma \ref{convergence lem}.
Therefore we have 
{\allowdisplaybreaks \begin{align*}
&\dfrac{e^{2\pi\ii\langle x, y\rangle}\prod_{j=1}^{n}\chi_j(-1)\mathcal{G}(\chi_j)}{|\det A|}
\times \sum_{m \in \Z^n}\bar{\chi}(-m) e^{-2\pi \ii\langle x, y-m\t A^{-1} \rangle} \\
& \times \bigg\{\prod_{j=1}^{n}\1_{[-1/2,1/2]}(y_j- (m \t A^{-1})_j)\exp(t_j\cos 2\pi(y_j- (m \t A^{-1})_j))\bigg\}
\notag \\
=& \dfrac{e^{2\pi\ii\langle y,x\rangle}\prod_{j=1}^n\chi_j(-1)\mathcal{G}(\chi_j)}{|\det A|}\times \dfrac{\prod_{j=1}^n\mathcal{G}(\bar \chi_j)}{q^n|\det\t A^{-1}|}\sum_{k\in\Z^n} \chi(k)
\bigg\{\prod_{j=1}^n I_{x_j+\tfrac{(k A)_j}{q}}(t_j)\bigg\}e^{2\pi \ii\langle y, \tfrac{k A}{q}\rangle},
\notag
\end{align*}
}where $(kA)_j$ and $(m\t A^{-1})_j$ are the $j$th entries of $kA$ 
and of $m\t A^{-1}$, respectively.
As we see $\chi_j(-1)\mathcal{G}(\bar\chi_j)\mathcal{G}(\chi_j)=|\mathcal{G}(\chi_j)|^2=q$,
we obtain \eqref{chi-I-E}.
\end{proof}
\begin{proof}[Proof of Corollary \ref{thm:renormalization}]
Now we apply Theorem \ref{thm:I-E:char} by substituting $L\Gamma$, $Lx$ and $y/L$ for $\Gamma$, $x$ and $y$, respectively.  
Noting $(L\Gamma)^\ast=\frac{1}{L}\Gamma^\ast$,
we have 
\begin{align*}
	&\sum_{k\in\Z^n} \chi(k)
\bigg\{\prod_{j=1}^n I_{\left(x_j+\tfrac{(kA)_j}{q}\right)L}(t_j)\bigg\}e^{2\pi \ii\langle \tfrac{y}{L},\big(x+\tfrac{kA}{q}\big)L\rangle}\\
= &
	\dfrac{\prod_{j=1}^n\mathcal{G}(\chi_j)}{L^n|\det A|}
	\sum_{m \in \Z^n}\bar{\chi}(m) \times e^{2\pi \ii\langle Lx, \frac{m\t A^{-1}}{L} \rangle}\notag\\
& \times \bigg\{	\prod_{j=1}^{n}\1_{[-1/2,1/2]}\left( \frac{y_j-(m\t A^{-1})_j}{L}\right)\exp\left({t_j\cos \frac{2\pi\left(y_j-  (m\t A^{-1})_j\right)}{L}}\right)\bigg\}.
\end{align*}
By multiplying the both sides with $\prod_{j=1}^n \sqrt{t_j}e^{-t_j}$ 
and replacing $t_j$ with $L^2t_j$ for any $L\in \ZZ_{\ge 1}$, we have the equation
{\allowdisplaybreaks \begin{align}\label{chi-I-E-L}
	&\sum_{k \in \Z^n} \chi(k)
	\bigg\{\prod_{j=1}^n L\sqrt{t_j}e^{-L^2t_j} I_{\left(x_j+\tfrac{(kA)_j}{q}\right)L}(L^2t_j)\bigg\}
	e^{2\pi \ii\langle \tfrac{y}{L},\big(x+\tfrac{kA}{q}\big)L\rangle}
	\\
	=&
	\dfrac{\prod_{j=1}^n\mathcal{G}(\chi_j)}{|\det A|}
	\sum_{m \in \Z^n}
	\bar{\chi} (m) \times e^{2\pi \ii\langle Lx,\frac{m\,\t A^{-1}}{L} \rangle} \notag\\
	&\times
	\bigg\{\prod_{j=1}^{n}\sqrt{t_j}\,e^{-L^2t_j}\1_{[-1/2,1/2]}\left(\frac{y_j- (m\,\t A^{-1})_{j}}{L}\right)\exp\left(L^2t_j\cos \frac{2\pi(y_j-(m\, \t A^{-1})_{j})}{L}\right)
	\bigg\}.
	\notag
\end{align}
}Suppose $t_1,\dots,t_n\in \R_{> 0}$.
By taking the limit as $L\to\infty$
and applying the Lebesgue dominated convergence theorem
with the aid of Lemma \ref{lem:L-sum},
we have the following identity of theta series with characters: 
\begin{align*}
	&\frac{1}{(2\pi)^{n/2}}\sum_{k\in\Z^n} \chi(k)e^{-\frac{1}{2}\sum_{j=1}^{n}\frac{\left(x_j+\tfrac{(kA)_j}{q}\right)^2}{t_j}}e^{2\pi\ii\langle y, x+\tfrac{kA}{q}\rangle}
	\\
	&=
	\dfrac{\prod_{j=1}^n\mathcal{G}(\chi_j)}{|\det A|}
	\sum_{m\in \Z^n}\bar{\chi}(m) 
	\{\prod_{j=1}^{n}\sqrt{t_j} \}\exp(-2\pi^2 \sum_{j=1}^n(y_j-(m\,\t A^{-1})_j)^2 t_j+2\pi \ii\langle x,m\,\t A^{-1} \rangle),
	\notag
\end{align*}
where we use \eqref{limit-I}.
Thus we obtain \eqref{chi-I-E-infty}.
\end{proof}

	We derived \eqref{chi-I-E-infty} from \eqref{chi-I-E} by taking the limit $L \rightarrow \infty$.
We call the manipulation
the \textit{continuum limit} of \eqref{chi-I-E}.

\begin{remark}
As for the equality \eqref{chi-I-E-infty}, we consider the case of $q=1$.
Set
$\Omega=\,\t A^{-1}\diag(2\pi\ii t_1, \ldots, 2\pi\ii t_n)A^{-1}.$
Then the equality \eqref{chi-I-E-infty} is written as
\begin{align*}
		& \sum_{m \in \Z^n}\exp\left(\pi\ii\,(xA^{-1}+m)(-\Omega^{-1})\, {}^{t}(xA^{-1}+m)+2\pi \ii\, y\,\t A\, {}^{t}m \right)\exp(2\pi\ii y\,{}^{t}x)\\
		= & \det(\Omega/\ii)^{1/2}
		\sum_{m \in \Z^n}\exp\left(\pi\ii \,(y\,\t A+m)\Omega\,{}^{t}(y\,\t A+m) -2\pi\ii xA^{-1}\,{}^{t}m\right).
\end{align*}
Here we note that our vectors are row vectors.
This equation above also follows from the functional equation of the Riemann theta function
$\vartheta\left[\begin{matrix}n_1 \\ n_2
\end{matrix}\right](z,\Omega)$ with characteristics for the action of
$[\begin{smallmatrix}
	0_n & -1_n\\
	1_n & 0_n
\end{smallmatrix}]$, where $1_n$ is the $n$-by-$n$ unit matrix, $0_n$ is the $n$-by-$n$ zero matrix, $n_1=y\,\t A$ and $n_2=-xA^{-1}${\rm:}
	$$\vartheta\left[\begin{matrix}-n_2 \\ n_1
	\end{matrix}\right](0, -\Omega^{-1})=\det(\Omega/\ii)^{1/2}
\exp(-2\pi\ii n_1\,{}^{t}n_2)\vartheta\left[\begin{matrix}
	n_1\\ n_2
\end{matrix}\right](0,\Omega)$$
{\rm(}see \cite[(5.2)]{Mumford} for $[\begin{smallmatrix}
	A & B\\
	C & D
\end{smallmatrix}]=[\begin{smallmatrix}
0_n & -1_n\\
1_n & 0_n
\end{smallmatrix}]${\rm)}.
Here the vectors $n_1$ and $n_2$ are elements of $\R^n$ in contrast to \cite{Mumford}.
Furthermore, \eqref{chi-I-E-infty} for $q=1$ is equivalent to \cite[(5.2)]{Mumford}
for $[\begin{smallmatrix}
	A & B\\
	C & D
\end{smallmatrix}]=[\begin{smallmatrix}
0_n & -1_n\\
1_n & 0_n
\end{smallmatrix}]$
with the aid of the identity theorem for functions of several complex variables.
\end{remark}

\begin{remark}\label{trace formula for torus}
Note that the transformation formula such as \eqref{chi-I-E-infty} of the theta series of a lattice $\Gamma$  
can be seen as 
the trace formula of the Laplacian of the real torus $\R^n/\Gamma$
{\rm(}e.g., \cite{McKean, McKeanC, Sunada1980, Sunada1981}{\rm)}.
Similarly, by considering Remark \ref{rem:partial trace}, the right-hand side of \eqref{chi-I-E-L} for $y=0$ and $q=1$
is the sum of eigenfunctions of the Laplacian $\Delta'$ corresponding to 
the subset of eigenvalues given by
$$
\left\{
L^2\sum_{j=1}^n \left(\cos 2\pi  \frac{\gamma_j^\ast}{L}-1\right)
\, \bigg| \,
\frac{\gamma^\ast}{L}
=\left(\frac{\gamma_1^\ast}{L}, \dots, \frac{\gamma_n^\ast}{L}\right) \in \frac{1}{L} \Gamma^\ast \cap [-1/2,1/2]^n
\right\}.
$$
Here $\Delta'$ 
is defined by \eqref{laplacian'}.
Set
$$a_{L,\gamma^*} := \begin{cases}L^2\sum_{j=1}^n \left(\cos 2\pi  \frac{\gamma_j^\ast}{L}-1\right) & (\gamma^* \in \Gamma^* \cap [-L/2, L/2]^n),\\
0 & (\gamma^* \in \Gamma^*-[-L/2, L/2]^n).
\end{cases}$$
Then we have $\lim_{L\rightarrow \infty} a_{L, \gamma^*}=-2\pi^2\|\gamma^*\|$,
and the multiset
$$\{ -2\pi^2 \| \gamma^\ast \|^2  \mid \gamma^\ast \in \Gamma^\ast \}
$$
consists of all eigenvalues of $\frac{1}{2}\Delta_{\R^n/\Gamma}$, where $\Delta_{\R^n/\Gamma}$ is the Laplacian of $\R^n/\Gamma$ {\rm(}e.g., \cite[Proposition 3.14]{Sakai}{\rm)}.
\end{remark}
We present several examples of Theorem \ref{thm:I-E:char} and Corollary \ref{thm:renormalization}.
\begin{example}\label{ex:eta}
First, we give an identity of theta series with characters. 
Let $x$ be an integer and let $a$ and $q$ be positive integers with $q\ge 2$. 
Let $\chi : (\Z/q\Z)^\times \to \C^\times$ be a primitive Dirichlet character modulo $q$. Suppose that $q | a$. 
Applying Theorem \ref{thm:I-E:char} to the case where $n=1$, $\Gamma= a \Z$, $\Gamma^\ast= \frac{1}{a} \Z$ and $y=0$, we obtain 
\begin{align}\label{ex:1dim:chi}
\sum_{k \in \Z} \chi(k) I_{x+ a \tfrac{k}{q}}(t) =  \dfrac{\mathcal{G}(\chi)}{a} \sideset{}{'} \sum_{j=-[a/2]}^{[a/2]} 
\bar{\chi}(j) \exp(t \cos \tfrac{2\pi j}{a}) e^{2\pi \ii x \tfrac{j}{a}}.
\end{align}
Let $b$ and $L$ be positive integers. 
Take $\ell \in \frac{1}{2} \Z_{\geq 1}$ and $\alpha>0$, and 
suppose $q | 2 \ell$. 
When $x= bL$, $a= 2\ell L$ and $t= \dfrac{2L^2z}{\alpha^2}$ in \eqref{ex:1dim:chi},
we have 
\begin{align}\label{chi-I}
\sum_{k \in \Z} \chi(k)I_{\left(2\ell \tfrac{k}{q}+b\right) L} \left( \tfrac{2L^2z}{\alpha^2} \right) 
= \dfrac{\mathcal{G}(\chi)}{2\ell L} \sideset{}{'} \sum_{j=-\ell L}^{\ell L} 
\bar{\chi}(j) \exp 
\left(
\tfrac{2L^2z}{\alpha^2} \cos \tfrac{\pi j }{\ell L}
\right)
e^{\pi \ii b \tfrac{j}{\ell}}. 
\end{align}
By multiplying the both sides by $L e^{-\frac{2L^2z}{\alpha^{2}}}$
and taking the limit as $L\to\infty$, with the aid of \eqref{limit-I} 
and Lemma \ref{lem:L-sum}, 
we have the following identity of theta series with characters{\rm :}
\begin{align}\label{chi-theta}
\dfrac{\alpha}{\sqrt{4 \pi z}} \sum_{k \in \Z} \chi(k) 
e^{-\tfrac{\alpha^2 (2 \ell k/q+ b)^2}{4z}} 
= \dfrac{\mathcal{G}(\chi)}{2 \ell}
\sum_{j \in \Z} \bar{\chi}(j) e^{-\tfrac{\pi^2 j^2 z}{\alpha^2 \ell^2}} 
e^{\pi \ii b \frac{j}{\ell}}. 
\end{align}
Here we use 
\begin{align*} 
e^{-\frac{2L^2z}{\alpha^{2}}} \exp \left( \tfrac{2L^2z}{\alpha^2} \cos \tfrac{\pi j }{\ell L} \right) 
&= \exp \left( \tfrac{2L^2z}{\alpha^2} \left( \cos \tfrac{\pi j }{\ell L}- 1 \right) \right) \\
&= \exp \left( \tfrac{-4 L^2z}{\alpha^2} \sin^{2} \tfrac{\pi j }{2\ell L} \right) 
\to e^{-\tfrac{\pi^2 j^2 z}{\alpha^2 \ell^2}}, \qquad (L\rightarrow \infty). 
\end{align*}
This manipulation from \eqref{ex:1dim:chi} to \eqref{chi-theta} is a continuum limit.

Next we give a transformation identity of the Dedekind eta function. 
Let $(\tfrac{12}{\cdot})$ be the unique primitive Dirichlet character modulo $12$. 
By specializing our case to $\chi(k)= (\tfrac{12}{k})$, $b=0$, $\alpha=1$, $q=12$, $\ell=6$ and $z=3t$ in \eqref{chi-theta}, 
we have 
\begin{align*}
\dfrac{1}{\sqrt{12 \pi t}} \sum_{k \in \Z} \left( \frac{12}{k} \right) e^{-\frac{k^2}{12t}} 
= \dfrac{\mathcal{G}((\tfrac{12}{\cdot}))}{12} 
\sum_{j \in \Z} \left( \frac{12}{j} \right) e^{-\frac{\pi^2 j^2t}{12}}. 
\end{align*}
By using $\mathcal{G}((\tfrac{12}{\cdot}))=\sqrt{12}$
and setting $t=\frac{\tau}{\ii\pi}$ in the identity above, 
we obtain 
$$
\sqrt{\frac{i}{\tau}} \sum_{k \in \Z} \left( \frac{12}{k} \right) 
e^{2\pi \ii (\tfrac{-1}{\tau}) \tfrac{k^2}{24}} 
= 
\sum_{j \in \Z} \left( \frac{12}{j} \right) e^{2 \pi \ii \tau \tfrac{j^2}{24}}. 
$$
As the Dedekind eta function $\eta$ can be written as 
$\eta(\tau)= \frac{1}{2} \sum_{n \in \Z} \left( \frac{12}{n} \right) 
e^{2 \pi \ii \tau \frac{n^2}{24}}$, 
finally we obtain the transformation formula of $\eta(\tau)${\rm:}
$$
\sqrt{\frac{\ii}{\tau}} \eta(-1/\tau)= \eta(\tau).
$$
\end{example}
\begin{example}
We see certain transformation formulas of theta  functions 
which can be obtained by \eqref{I-E} and \eqref{limit-I}. 
Take $L \in \ZZ_{\ge 1}$.  
By Theorem \ref{thm:I-E:char} for $x=L$, $a=2L$ and $q=1$, we have 
$$
\sum_{k \in \Z}  I_{(2k+1) L}(t)
=
\dfrac{1}{2L}
\sideset{}{'}\sum_{j=-L}^{L}
(-1)^j
\exp\left(t\cos \frac{\pi j}{L}\right).
$$
By replacing $t$ with $2L^2t$, multiplying the both sides by $Le^{-2L^2t}$
and taking the limit as $L\to\infty$, we have
\begin{align}\label{theta-identity}
\dfrac{1}{\sqrt{4\pi t}}\sum_{k\in\Z}e^{-\frac{(2k+1)^2}{4t}}
=\dfrac{1}{2}\sum_{j\in\Z}(-1)^j e^{-\pi^2 j^2 t}
\end{align}
with the aid of \eqref{limit-I}.
Indeed, the right-hand side is deduced from
\begin{align*}
&\dfrac{1}{2}
\sideset{}{'}\sum_{j=-L}^{L}
(-1)^j
\exp\left(2L^2t\left(\cos \frac{\pi j}{L}-1\right)\right)
=
\dfrac{1}{2}
\sideset{}{'}\sum_{j=-L}^{L}
(-1)^j
\exp\left(-4L^2t\cdot \sin^2 \frac{\pi j}{2L}\right)
\\
&=
\dfrac{1}{2}
\sum_{j\in \Z}
\1_{[-L,L]}(j)(-1)^j
\exp\left(-\pi^2 j^2t\left(\frac{\sin \frac{\pi j}{2L}}{\frac{\pi j}{2L}}\right)^2\right)
\to
\dfrac{1}{2}\sum_{j\in\Z}(-1)^j e^{-\pi^2 j^2 t}
\end{align*}
as $L\rightarrow \infty$.
The identity \eqref{theta-identity} is a well-known relation of the theta functions explained as follows.  
Following the notation of \cite[\S 76]{Rademacher}, we set
$$
\vartheta_2(v|\tau):=\sum_{n\in\Z}e^{\pi \ii \tau (n+1/2)^2}e^{(2n+1)\pi \ii v},
\quad 
\vartheta_4(v|\tau):=\sum_{n\in\Z}(-1)^n e^{\pi \ii \tau n^2}e^{2\pi \ii nv}
$$
for any $v,\tau \in \C$ such that $\Im\,\tau>0$. 
By substituting $t=\frac{\tau}{\ii\pi}$, the identity \eqref{theta-identity} is rewritten as the following {\rm (}cf.\ \cite[(79.9)]{Rademacher}{\rm ):}
$$
\dfrac{1}{2}\sqrt{\dfrac{\ii}{\tau}}\,\vartheta_2(0|-\tfrac{1}{\tau})
=
\dfrac{1}{2}\,\vartheta_4(0|\tau).
$$
This is a specialization of Jacobi's imaginary transformation 
$$\vartheta_2(\tfrac{v}{\tau}|-\tfrac{1}{\tau})
=
e^{-\pi \ii /4}\sqrt{\tau}e^{\pi \ii v^2/\tau}\vartheta_4(v|\tau).
$$
\end{example}

\section{Lattice sums of $I$-Bessel functions coming from linear codes}
\label{Lattice sums of I-Bessel functions coming from linear codes}

In this section, we express $I$-Bessel lattice sums 
where the lattice involved is defined by a linear code over a finite commutative ring.
For general references of theory of linear codes over rings,
see \cite{NRS,Dougherty}. 
Let $m$ be an integer with $m \ge 2$. 
Now consider the quotient ring $R:= \Z/m\Z$. 
For $N \in \Z_{\ge 1}$, let $\rho_N: \Z^N \to R^N$ be the group homomorphism 
defined by the reduction modulo $m$, that is, $\rho_N((a_j)_{j=1}^{N})=(a_j \pmod m)_{j=1}^{N}$.
We omit the subscript $N$ of $\rho_N$ if there is no confusion. 
Let $n$ be a positive integer.
A {\it linear code $C$ over $R$ of length $n$} is an $R$-submodule of $R^n$. 
The dual code $C^\bot$ is defined by 
$$
C^\bot := \{ x \in R^n \mid \mbox{ $x \cdot y= 0$ for all $y\in C$} \}, 
$$
where $x \cdot y:= \sum_{i=1}^n x_i y_i$ for $x, y \in C$. 
We remark $\#R^n=\#C \times \#C^{\bot}$ (cf.\ \cite[Proposition 7.6 and Theorem 7.7]{Wood}).
For a subset $B$ of $R^n$, the {\it complete weight enumerator}
of $B$ is a polynomial of $X_0, \dots, X_{m-1}$ defined by 
$$
\cwe_B (X_0,\dots, X_{m-1}):=\sum_{c \in B}\prod_{j=0}^{m-1} X_j^{n_j(c)} \in \Z[X_0, X_1, \ldots, X_{m-1}],
$$ 
where $n_j(c)$ is the number of components of $c$ that are congruent to $j$ (mod $m$). 
This polynomial $\cwe_B (X_0,\dots, X_{m-1})$ is homogeneous, 
that is, 
$$
\cwe_B (t X_0, \dots, t X_{m-1})= t^{n} \cdot \cwe_B (X_0, \dots, X_{m-1}). 
$$
If $X= X_0$ and $Y= X_1=\cdots =X_{m-1}$, 
then $W_B(X, Y):=\cwe_B(X,Y, \ldots,Y)=$ is called the {\it weight enumerator} of $B$.
The weight enumerator $W_B(X,Y)$ is expressed as
$$
W_B(X,Y)=\sum_{c \in B}X^{n-\wt(c)}Y^{\wt(c)},
$$
where $\wt(c):= \# \{i \in \{1,\ldots, n\} \mid c_i \not \equiv 0 \pmod m \}$ is 
called the {\it Hamming weight} of $c$.
If $m$ is a prime number, then $R=\Z/m\Z$ is a field and hence
$C$ is a free $R$-module. 
In this case, we can define the dimension of $C$.
An {\it $[n,k]$-linear code} means a linear code whose length and dimension 
are $n$ and $k$, respectively. 

Since the kernel of $\Z^{n} \to R^{n} \to R^{n}/C$ is $\rho^{-1}(C)$, 
we obtain a group isomorphism $\Z^n/\rho^{-1}(C) \cong R^n/C$. 
Hence, $\rho^{-1}(C)$ is a free $\Z$-submodule with finite index $\# R^n/\# C$ in $\Z^n$. 
Note $\rank_\Z \rho^{-1}(C)=n$ and $\vol(\R^n/\rho^{-1}(C))= \# R^n/\# C$.
By Theorem \ref{thm:I-E}, we have
\begin{align}\label{Identity from Possion}
& \sum_{\gamma \in x+ \rho^{-1}(C)} \prod_{j=1}^n I_{\gamma_j}(t_j) \\
= &
\frac{\#C}{\#R^n}
 \sum_{\gamma^\ast \in (\rho^{-1}(C))^\ast}
\left\{ \prod_{j=1}^{n} \1_{[-1/2,1/2]}(\gamma_j^*)
\exp(t_j \cos 2 \pi \gamma_j^\ast)\right\} e^{2 \pi \ii\langle x, \gamma^\ast \rangle} \notag
\end{align}
for $x \in \ZZ^n$.
\begin{remark}
	An $I$-Bessel lattice sum defined by a system of some linear congruence conditions can be written  as a sum over a lattice 
	related to a linear code.
	We explain it as follows.
	Take an integer $m\ge 2$ and a $k\times n$ matrix $H$ whose components are contained in $\ZZ$.
	Let $C$ be the linear code over $R$
	defined by  
	$$
	C=\{ c \in R^n \mid \rho (H)\ ^t c =0 \},
	$$ 
	where $\rho(H)$ is the reduction of $H$ modulo $m$.
	The matrix $\rho(H)$ is called a {parity-check matrix} of $C$.
	 We obtain 
	\begin{align*}
		\sum_{\substack{\gamma\in \ZZ^n \\ H ^t\gamma\equiv 0 \textrm{ {\rm(mod} $m${\rm)}}}} \prod_{j=1}^n I_{\gamma_j}(t_j)
		=
		\sum_{\gamma \in \rho^{-1}(C)} \prod_{j=1}^n I_{\gamma_j}(t_j).
	\end{align*}
\end{remark}
In the following, we give other expressions of the aforementioned $I$-Bessel lattice sum 
(i.e., the left hand-side of \eqref{Identity from Possion}) in terms of the complete weight enumerator of 
a linear code $C$.
\begin{lemma}\label{I-cwe}
Let $C$ be a linear code over $R=\ZZ/m\ZZ$ of length $n$ with $m\ge2$.
Take $x= (x_{1}, \ldots, x_{n}) \in \ZZ^n$ and $t_1, \dots, t_n \in \C$. 
Then, we obtain 
\begin{align}\label{I-A}
\sum_{\gamma \in x+ \rho^{-1}(C)} \prod_{j=1}^n {I}_{\gamma_j}(t_j)
=
\sum_{c \in C}
\prod_{i=0}^{m-1}
\prod_{1 \le j \le n \atop c_j = \rho(i)} A_{x_j+ i}(t_j),
\end{align}
where $\gamma= (\gamma_{1}, \ldots, \gamma_{n})$, $c= (c_{1}, \ldots, c_{n})$ and 
$$
A_y(t)
:=
\sum_{\gamma \in y+ m\Z} {I}_{\gamma }(t)
= \dfrac{1}{m}
\sideset{}{'} \sum_{j= -[m/2]}^{[m/2]} \exp\left(t \cos \frac{2 \pi j}{m}\right) e^{2 \pi \ii y \frac{j}{m}},
\qquad y \in \Z.$$ 

In particular, for $x\in \Z^n$ and $t \in \C$ we have
\begin{align}\label{I-A-Z}
\sum_{\gamma \in x+ \rho^{-1}(C)} \prod_{j=1}^n I_{\gamma_j}(t)
= \cwe_{\rho(x)+ C}(A_0(t),\dots, A_{m-1}(t)).
\end{align}

\end{lemma}
\begin{proof}
For a vector $v=(v_1,\dots,v_n)\in\R^n$ and a subset $S\subset \{1,\dots, n\}$, 
let $v^S \in \R^{n-\#S}$ denote the following punctured vector: 
$$
v^S:=( v_j  )_{j \in \{1,\dots,n\} -S}.
$$
For example, we have $v^{\{1\}}= (v_{2}, \ldots, v_{n})$ and $v^{\{1,2\}}= (v_{3}, \ldots, v_{n})$. 

Let $x= (x_{1}, \ldots, x_{n} ) \in \Z^n$ and $c= (c_{1}, \ldots, c_{n}) \in C$, and fix them. 
Then we have
{\allowdisplaybreaks \begin{align*}
&\sum_{\gamma \in x+ \rho^{-1}(c)} \prod_{j=1}^n {I}_{\gamma_j}(t_j) 
\\
= &
\left(
\sum_{\gamma_1 \in x_1+ \rho^{-1}(c_1)} {I}_{\gamma_1}(t_1 )
\right)
\left(
\sum_{\gamma^{\{ 1 \}} \in x^{\{ 1 \}} + \rho^{-1}(c^{\{ 1 \}}) }
\prod_{j=2}^n {I}_{\gamma_j}(t_j)
\right)
\\
= &
\left(
\sum_{\gamma_1 \in x_1+ \rho^{-1}(c_1)} {I}_{\gamma_1}(t_1)
\right)
\left(
\sum_{\gamma_2 \in x_2+ \rho^{-1}(c_2)} {I}_{\gamma_2}(t_2)
\right) 
\\
& \times 
\left(
\sum_{\gamma^{\{ 1,2 \}} \in x^{\{ 1,2 \}} + \rho^{-1}(c^{\{ 1,2 \}})}
\prod_{j=3}^n {I}_{\gamma_j }(t_j)
\right)
\\
= &
\cdots
\\
= &
\prod_{j=1}^n
\left(
\sum_{
\gamma_j \in x_j+ \rho^{-1}(c_j)} {I}_{\gamma_j}(t_j)
\right)
\\
= &
\prod_{j=1}^n
A_{x_j}(t_j)^{\delta(\rho^{-1}(c_j)= m\Z)}
A_{x_j+ 1}(t_j)^{\delta(\rho^{-1}(c_j)= 1+ m\Z)}
\cdots
A_{x_j+ m-1}(t_j)^{\delta(\rho^{-1}(c_j)= m-1+ m\Z)}
\\
= &
\prod_{j=1}^n
A_{x_j}(t_j)^{\delta(c_j= 0)}
A_{x_j+ 1}(t_j)^{\delta(c_j= 1)}
\cdots
A_{x_j+ m-1}(t_j)^{\delta(c_j= m-1)}\\
= & 
\prod_{i=0}^{m-1}\prod_{1 \le j \le n \atop c_j = \rho(i)} A_{x_j+ i}(t_j), 
\end{align*}
}where, for a condition $\rm P$, $\delta(\rm P)$ denotes the generalized Kronecker symbol so that $\delta(\rm P)=1$ if
$\rm P$ is true, and $\delta(\rm P)=0$ if $\rm P$ is false, respectively.
Hence we have \eqref{I-A}. 

Next we suppose $t_1= \cdots= t_n= t$.
We have the identities
$$
A_y(t) = 
\sum_{\gamma \in y+ m\Z} {I}_{\gamma }(t)
= \dfrac{1}{m} \sum_{j= 0}^{m-1} \exp\left(t \cos \frac{2 \pi j}{m}\right) e^{2 \pi \ii y \frac{j}{m}}
$$
and $A_{y}= A_{y+m\ell}$ for all $\ell \in \Z$. 
Combining these, we obtain
$$
\prod_{i=0}^{m-1} \prod_{1 \le j \le n \atop c_j= \rho(i)}A_{x_j+ i}(t)
=
\prod_{\ell= 0}^{m-1} A_{\ell}(t)^{\# \{1 \le j \le n \mid \rho(x_j)+ c_j= \rho(\ell) \}}
$$
for $x \in \Z^n$. 
This implies
\begin{align*}
\sum_{\gamma \in x+ \rho^{-1}(C)} \prod_{j=1}^n I_{\gamma_j}(t)
&= \sum_{c \in C} \prod_{\ell= 0}^{m-1} A_{\ell}(t)^{\# \{1 \le j \le n \mid \rho(x_j)+ c_j= \rho(\ell)\}}
\\
&= \cwe_{\rho(x)+C} (A_0(t), \dots, A_{m-1}(t)).
\end{align*}
Therefore \eqref{I-A-Z} follows and hence we are done.
\end{proof}

\begin{theorem}\label{thm:code-main}
For any $t \in \C$ and any $x \in \Z^n$, we have
\begin{align}\label{x-MacWilliams-I}
	&\cwe_{\rho(x)+ C}(mA_0(t),\dots, mA_{m-1}(t)) 
	=\#C\sum_{c\in C^\bot}\prod_{j=1}^{n}e^{t\cos(\frac{2\pi c_j}{m})} e^{\frac{2\pi\ii x_jc_j}{m}}.
\end{align}
In particular,
for any $t \in \C$ and any $x \in \Z$ we have
\begin{align}\label{diag-MacWilliams-I}
	&\cwe_{\rho(\iota(x))+C}(mA_0(t), \ldots, mA_{m-1}(t))\\
	=&\#C\, \cwe_{C^\bot}(e^{t\cos\frac{2\pi\cdot0}{m}}e^{2\pi\ii x\frac{0}{m}},\ldots,
	e^{t \cos\frac{2\pi \cdot (m-1)}{m}} e^{2\pi\ii x\frac{m-1}{m}}),
	\notag
\end{align}
where $\iota : \R \hookrightarrow \R^n$ is the diagonal embedding.
\end{theorem}

\begin{proof}
The left-hand side of \eqref{Identity from Possion} for $t=t_1=\cdots=t_n$ is given as
$$\cwe_{\rho(x)+C}(A_0(t), \ldots, A_{m-1}(t))$$
by \eqref{I-A-Z}.
We consider the right-hand side of \eqref{Identity from Possion}.

First we show $\rho^{-1}(C^{\bot})=m\{\rho^{-1}(C)\}^*$.
If $x \in \rho^{-1}(C^{\bot})$, we have $\rho(x)\cdot c=0$ for any $c \in C$,
which leads us $\langle x, y\rangle \in m\Z$ for all $y \in \rho^{-1}(C)$.
This implies the containment $\rho^{-1}(C^{\bot})\subset m\{\rho^{-1}(C)\}^{*}$.
By combining
$\vol(\R^n/\rho^{-1}(C^{\bot})) = \#R^n/\#C^{\bot}$ and
$$\vol(\R^n/m\{\rho^{-1}(C)\}^*) =  
\frac{m^n}{\vol(\R^n/\rho^{-1}(C))}
= \frac{m^n\#C}{\#R^n}=\#C$$
with $\#C^\bot=\frac{\#R^n}{\#C}$,
we obtain $\vol(\R^n/\rho^{-1}(C^{\bot})) = \vol(\R^n/m\{\rho^{-1}(C)\}^*)$.
Hence the proof of $\rho^{-1}(C^{\bot})=m\{\rho^{-1}(C)\}^*$ is complete.

Next we transform the right-hand side of \eqref{Identity from Possion} for $t=t_1=\cdots=t_n$
into
\begin{align*}
& \frac{\#C}{\#R^n} \sum_{\gamma^{*}\in \rho^{-1}(C^\bot)}\prod_{j=1}^{n}1_{[-1/2,1/2]}\left(\frac{\gamma_j^*}{m}\right)
\exp\left(t\cos\left(\frac{2\pi\gamma_j^{*}}{m}\right)\right) e^{2\pi\ii x_j \frac{\gamma_j^*}{m}} \\
= & \frac{\#C}{\#R^n}\sum_{c \in C^\bot}
\prod_{j=1}^{n}\exp\left(t\cos\left(\frac{2\pi c_j}{m}\right)+2\pi\ii x_j\frac{c_j}{m}\right) \notag\end{align*}
for $x \in \Z^n$.
This shows \eqref{x-MacWilliams-I}.

In particular, when there exists $x_0\in \Z$ such that $x_j=x_0$ for all $j\in\{1,\ldots, n\}$,
the right-hand side of the equality above is evaluated as
\begin{align}
& \frac{\#C}{\#R^n}\sum_{c \in C^\bot}
\prod_{\ell=0}^{m-1}\left\{\exp\left(t \cos\left(\frac{2\pi \ell}{m}\right)+2\pi\ii x_0\frac{\ell}{m}\right)\right\}^{n_{\ell}(c)} \notag \\
= & \frac{\#C}{\#R^n}\cwe_{C^\bot}(z_0,\ldots, z_{m-1}),\notag
\end{align}
where $z_\ell= \exp(t \cos(\frac{2\pi \ell}{m})+2\pi\ii x_0\frac{\ell}{m})$.
Therefore we obtain 
\begin{align*}
& \cwe_{\iota(x_0)+C}(\sum_{\ell=0}^{m-1}y_\ell, \sum_{\ell=0}^{m-1}\zeta^\ell y_{\ell},
\ldots, \sum_{\ell=0}^{m-1}\zeta^{(m-1)\ell}y_\ell )\\
= & \#C\, \cwe_{C^\bot}(y_0 e^{2\pi\ii x_0\frac{0}{m}},\ldots, y_{m-1}e^{2\pi\ii x_0\frac{m-1}{m}})
\end{align*}
with $\zeta=e^{\frac{2\pi\ii}{m}}$ and $y_\ell= \exp(t \cos(\frac{2\pi \ell}{m}))$.
This completes the proof of \eqref{diag-MacWilliams-I}.
\end{proof}

\begin{remark}
Wood \cite[Theorem 8.1]{Wood} proved the MacWilliams type identity for a linear code $C$ over a finite Frobenius ring $R$. If $R=\Z/m\Z$, it becomes 
\begin{align}\label{Wood identity}
	& \cwe_{\rho(\iota(x_0))+C}(\sum_{\ell=0}^{m-1}X_\ell, \sum_{\ell=0}^{m-1}\zeta^\ell X_\ell,
	\ldots, \sum_{\ell=0}^{m-1}\zeta^{(m-1)\ell}X_\ell )\\
	= & \#C\, \cwe_{C^\bot}(X_0 e^{2\pi\ii x_0\frac{0}{m}},\ldots, X_{m-1}e^{2\pi\ii x_0\frac{m-1}{m}}).\notag
\end{align}
Note that Wood proved a more general identity than \eqref{Wood identity} above.
Theorem \ref{thm:code-main} is a discrete analogue of 
the equation satisfying theta functions of linear codes. 
For $R=\Z/m\Z$ and $x=(0,\dots,0)$, the corresponding identities of theta functions become \eqref{Wood identity} by the substitution  
$$
X_j=\theta_{\frac{j}{\sqrt{m}}+\sqrt{m}\Z}(\tau), \quad (j=0,\dots,m-1),
$$
{\rm(}e.g., \cite[p.87]{Nishimura} when $m$ is a prime number{\rm)}.
Here $\theta_{\Gamma}(\tau)$ is the theta function of a lattice $\Gamma$ defined by 
$$
\theta_{y+\Gamma}(\tau):=\sum_{\gamma \in \Gamma} \exp(\pi \ii \tau \|y+\gamma\|^2), \qquad \tau \in \C, \ \Im(\tau)>0.
$$
For general references of theta functions of lattices  and 
linear codes, see \cite{NRS, CS, Ebeling}. 
\end{remark}

\begin{example}\label{ex:codes}
If $\Gamma=2\Z$, then we have $\Gamma^\ast= \frac{1}{2}\Z$ and $\vol(\R/\Gamma)=2$. 
For $x \in \Z$ and $t\in \C$, $A_x(t)$ is described as
\begin{align*}
A_x(t)=\sum_{\gamma \in x+2\Z} {I}_{\gamma }(t)
=
\dfrac{1}{2}\sideset{}{'}
\sum_{j=-[2/2]}^{[2/2]}\exp(t\cos \pi j )e^{2\pi \ii x \frac{j}{2}}
=\frac{e^t+(-1)^xe^{-t}}{2}.
\end{align*}

Let $C$ be an $[n,k]$-binary linear code, i.e., an $[n,k]$-linear code over $\Z/2\Z$.
By 
Theorem \ref{thm:code-main}, we obtain
\begin{align*}W_{C}(2A_x(t), 2A_{x+1}(t)) = & \cwe_{\rho(\iota(x))+C}(2A_{0}(t), 2A_{1}(t))
=\#C\cwe_{C^\bot}(e^t, (-1)^{x}e^{-t})\\
= &2^{k} W_{C^\bot}(e^t, (-1)^{x}e^{-t})
\end{align*}
for $x\in \Z$.
We remark that this identity also follows from the MacWilliams identity
\begin{align*}
W_{C^\bot}(X,Y)=\dfrac{1}{2^k}W_C(X+Y,X-Y).
\end{align*}
\end{example}

\section{Sums of $I$-Bessel functions and heat equations on lattices}
\label{Sums of I-Bessel functions and heat equations on lattices}

The continuous time discrete heat equation on $\Z^n$ 
has been studied by many researchers.
For the equation on $\Z$, the works \cite{GI, BBDS, CGRTV} discussed 
the fundamental solution and the heat semigroup.
For the equation on $\Z^n$, there are related works such as \cite{KN, CJK, CS, Yamasaki, Faustino}.
In particular, the works \cite{KN, CJK, Slavik, CJK2, Yamasaki} discussed 
heat equations on $\Z^n$ as an infinite graph.
There are many works related to heat equations on graphs.
To our best knowledge, the first research on heat equations on graphs 
is due to F.~R.~Chung and S.-T.~Yau (see \cite{Chung, Chung-Yau97, Chung-Yau2001} and 
the literatures cited there for details).

Let $\Lambda = \Z^n B$ be a lattice in $\R^n$.
We regard $\Lambda$ as a $2n$-regular graph.
For a bounded function $f : \Lambda \to \R$, 
let $\Delta_\Lambda$ denote the Laplacian defined by 
\begin{align}\label{laplacian}
\Delta_\Lambda f(x) := \dfrac{1}{2n} \sum_{y \sim x} \left( f(y)-f(x) \right)
=\frac{1}{2n}\sum_{j=1}^{n}(f(x+b_j)+f(x-b_j)-2f(x)), 
\end{align}
where $y$ runs over all vertices in $\Lambda$ 
such that $y$ is adjacent to $x$ in the $2n$-regular graph $\Lambda$, and
$b_j$ for each $j \in \{1,\ldots,n\}$ is the $j$th row vector of $B$, i.e., $B=\left(\begin{smallmatrix}b_1\\ \vdots \\ b_n\end{smallmatrix}\right)$.

The semidiscrete heat equations
(the continuous time discrete heat equations)
on discrete spaces have been studied (see \cite{GI, CS, Slavik, CGRTV, KLW}, etc.).
As for $\Lambda = \Z^n$, an explicit formula of the heat kernel $\K_{\Z^n, t}$ for $\Z^n$ is given as follows.

\begin{theorem}
[{\cite[\S2]{KN}, \cite[\S5.2]{Slavik}}]
	\label{th:heat:Z^n}
Assume $\Lambda = \Z^{n}$. 
Let $u_0 : \Lambda \to \R$ be a bounded function.
Then the heat equation 
	\begin{align}
		\begin{cases}
			\partial_t u(x,t) = \Delta_{\ZZ^n} u(x,t), & \qquad (x, t) \in \Z^n \times \R_{\ge0}, \\
			u(x,0) = u_0(x), & \qquad x \in \Z^n
		\end{cases}
	\end{align}
has a unique bounded solution $u : \Z^n \times \R_{\ge 0} \to \R$ which is differentiable in $t \in \R_{>0}$ and continuous in $t\in \R_{\ge0}$. Moreover it is
explicitly given by 
$$
u(x,t) = (u_0 * \K_{\Z^n,t})(x) 
:= \sum_{y \in \Z^n} u_0 (x-y) \K_{\Z^n,t}(y), 
$$
where $y= (y_{1}, y_{2}, \ldots, y_{n})$ and
$$
\K_{\Z^n ,t}(y) := e^{-t} \prod_{j=1}^n I_{y_j}(\tfrac{t}{n}). 
$$
\end{theorem}

The theorem above can be generalized for any lattice $\Lambda = \Z^{n} B$ which is not necessarily contained in $\Z^n$, and we give an explicit formula of the heat kernel
$\K_{\Lambda,t}$ for $\Delta_\Lambda$ as follows.

\begin{theorem}\label{th:heatLambda}
Let $\Lambda = \Z^{n} B$ be a lattice in $\R^n$
and $u_0 : \Lambda \to \R$ a bounded function.
Then, the heat equation
	\begin{align}\label{eq:heatLambda}
		\begin{cases}
			\partial_t u(x,t) = \Delta_\Lambda u(x,t), & \qquad (x, t) \in \Lambda \times \R_{\ge0}, \\
			u(x,0) = u_0(x), & \qquad x \in \Lambda
		\end{cases}
	\end{align}
has a unique bounded solution $u : \Lambda  \times \R_{\ge 0} \to \R$
which is differentiable in $t \in \R_{>0}$ and continuous in $t\in \R_{\ge0}$.
This is explicitly given by 
$$
u(x,t) = (u_0*\K_{\Lambda,t})(x):=\sum_{y \in \Lambda} u_0 (x-y) \K_{\Lambda, t}(y),
$$
where 
\begin{align*}
\K_{\Lambda, t}(y) := & \dfrac{1}{\vol(\R^n/\Lambda^*)}
e^{-t}
\int_{\R^n/\Lambda^*} e^{2 \pi \ii\langle \xi, y \rangle}
\exp\left(\frac{t}{n} \sum_{j=1}^n \cos 2 \pi \langle \xi, b_j \rangle \right) \, d\xi \\
= & e^{-t}
\prod_{j=1}^{n}I_{(yB^{-1})_j}(\tfrac{t}{n}), \qquad (y \in \Lambda, t \in \R_{\ge0}),
\end{align*}
$b_j$ is the $j$th row of $B$, and $(yB^{-1})_j$ is the $j$th component of the vector $yB^{-1}$.
\end{theorem}
\begin{proof}
This theorem can be proved simply by using Theorem \ref{th:heat:Z^n} and Lemma \ref{lem:Ftrans-I} if we set $\K_{\Lambda,t}(y):=\K_{\Z^n,t}(yB^{-1})$.
However, we give a proof of our explicit formula of $\K_{\Lambda,t}(y)$ without Theorem \ref{th:heat:Z^n}
in order to make our proof self-contained.

We prove the identity between $\K_{\Lambda,t}(y)$
	and the product of $I$-Bessel functions in the assertion as follows.
Let $\tilde b_j$ for each $j \in \{1,\ldots,n\}$ be the $j$th row vector of ${}^{t}B^{-1}$, i.e.,
${}^{t}B^{-1}=\left(\begin{smallmatrix}\tilde b_1\\ \vdots \\ \tilde b_n\end{smallmatrix}\right)$.
We put
$$\1_{\R^n/\Lambda^*}\left(\sum_{j=1}^{n}c_j\tilde b_j\right):=\prod_{j=1}^{n}\1_{[-1/2,1/2]}(c_j), \qquad (c_1,\ldots, c_n) \in \R^n,$$
where ${\bf1}_{[-1/2,1/2]}$ is the rectangular function given as \eqref{ch}.
Then $\K_{\Lambda,t}(y)$ is transformed into
{\allowdisplaybreaks\begin{align*}
	 &e^{-t}\int_{\R^n} \dfrac{1}{\vol(\R^n/\Lambda^*)}e^{2\pi\ii\langle\xi, y\rangle}\1_{\R^n/\Lambda^*}(\xi)\exp\left(\frac{t}{n} \sum_{j=1}^n \cos 2 \pi \langle \xi, b_j \rangle\right) d\xi \\
= & e^{-t}\int_{\R^n} e^{2\pi\ii\langle x{}^{t}B^{-1}, y\rangle}\1_{\R^n/\Lambda^*}(x{}^{t}B^{-1})\exp\left(\frac{t}{n} \sum_{j=1}^{n} \cos 2 \pi \langle x{}^{t}B^{-1}, b_j\rangle\right) dx \,  \\
= & e^{-t}\int_{\R^n}e^{2\pi\ii\langle x, yB^{-1}\rangle}
\left\{\prod_{j=1}^{n}{\bf1}_{[-1/2,1/2]}(x_j)
\exp\left(\frac{t}{n} \cos 2 \pi x_j \right) \right\}dx \\
= & e^{-t}
\prod_{j=1}^{n}\int_{\R}e^{2\pi\ii x_j(yB^{-1})_j}{\bf1}_{[-1/2,1/2]}(x_j)
\exp\left(\frac{t}{n} \cos 2 \pi x_j\right) dx_j\\
= & e^{-t}
\left\{\prod_{j=1}^{n}I_{(yB^{-1})_j}(\tfrac{t}{n})\right\},
\end{align*}
}where we use Lemma \ref{lem:Ftrans-I} to give the expression in the last line.
By Theorem \ref{thm:I-E} for $\Gamma=\ZZ^n$,  $x=y=0$ and $t_j=\frac{t}{n}$,
we obtain
\begin{align*}
\sum_{y \in \Lambda}\K_{\Lambda,t}(y)
= & \sum_{\gamma \in \ZZ^n}e^{-t}
\prod_{j=1}^{n}I_{\gamma_j}(\tfrac{t}{n}) 
= \sum_{\gamma^* \in \ZZ^n}e^{-t} \prod_{j=1}^{n}{\bf1}_{[-1/2,1/2]}(\gamma_j^*)
e^{\frac{t}{n} \cos 2 \pi \gamma_j^*}
=1.
\end{align*}
First we prove the existence of bounded solutions.
For any fixed bounded function $u_0:\Lambda\rightarrow \R$,
we prove that $u(x,t):=(u_0*\K_{\Lambda,t})(x)$
is a bounded solution to \eqref{eq:heatLambda}.
The defining series is absolutely convergent by
$$\sum_{y \in \Lambda}|u_0(x-y)||\K_{\Lambda,t}(y)|\le\sup_{z \in \Lambda}|u_0(z)|\sum_{y\in\Lambda}\K_{\Lambda,t}(y)=\sup_{z \in \Lambda } |u_0(z)|<\infty,$$
where we use $I_{(yB^{-1})_{j}}(\tfrac{t}{n})\ge 0$. In particular, $u$ is bounded.
The derivative $\partial_t \K_{\Lambda, t}(y)$ is computed as
\begin{align*}
	\partial_t \K_{\Lambda, t}(y)= &
	\dfrac{e^{-t}}{\vol(\R^n/\Lambda^*)}\int_{\R^n/\Lambda^*}
	e^{2\pi \ii\langle\xi,y\rangle}\left\{\frac{1}{n}\sum_{j=1}^{n}(\cos2\pi\langle\xi, b_j\rangle-1)\right\}\\
	& \times \exp\left(\frac{t}{n} \sum_{j=1}^n \cos 2 \pi \langle \xi, b_j \rangle\right) d\xi,
\end{align*}
where we exchange $\partial_t$ and the integral by the compactness of the domain.
The equality
$$e^{2\pi \ii\langle\xi,y\rangle}(\cos2\pi\langle\xi, b_j\rangle-1)
=\frac{1}{2} (e^{2\pi \ii\langle\xi,y+b_j\rangle}
+e^{2\pi \ii\langle\xi,y-b_j\rangle}-2e^{2\pi \ii\langle\xi,y\rangle})$$
leads us to
$\partial_t\K_{\Lambda,t}(y)=\Delta_\Lambda \K_{\Lambda, t}(y)$,
and hence
\begin{align*} & \Delta_\Lambda u(x,t)=\Delta_\Lambda (u_0*\K_{\Lambda, t})(x)\\
	=\, & \frac{1}{2n}\Bigg\{\sum_{y \in \Lambda}u_0(x+b_j-y)\K_{\Lambda,t}(y)
	+\sum_{y \in \Lambda}u_0(x-b_j-y)\K_{\Lambda,t}(y) \\
	& \qquad -2\sum_{y \in \Lambda}u_0(x-y)\K_{\Lambda,t}(y)\Bigg\}\\
	=\, & \frac{1}{2n}\Bigg\{\sum_{y' \in \Lambda}u_0(x-y')\K_{\Lambda,t}(y'+b_j)
	+\sum_{y' \in \Lambda}u_0(x-y')\K_{\Lambda,t}(y'-b_j)\\
& \qquad -2\sum_{y \in \Lambda}u_0(x-y)\K_{\Lambda,t}(y) \Bigg\}\\
	=\, & (u_0*\Delta_\Lambda \K_{\Lambda, t})(x) = (u_0*\partial_t\K_{\Lambda,t})(x)
	=\partial_t (u_0*\K_{\Lambda,t})(x)=\partial_t u(x,t).
\end{align*}
Here we implicitly exchange $\sum_{y\in\Lambda}$ and $\partial_t$ for proving
$(u_0*\partial_t\K_{\Lambda,t})(x)
=\partial_t (u_0*\K_{\Lambda,t})(x)$ above,
and this is justified by the locally uniform convergence of $(u_0*\partial_t\K_{\Lambda,t})(x)$ on $\R_{\ge 0}$ as a function in $t$.
Indeed, the locally uniform convergence follows from
\begin{align*}|(u_0*\partial_t\K_{\Lambda,t})(x)|=& |(u_0*\Delta_\Lambda \K_{\Lambda, t})(x)|
	\le \sup_{z \in \Lambda}|u_0(z)|\times\frac{1}{2n}\times 4 \sum_{y \in \Lambda}\K_{\Lambda,t}(y)
\end{align*}
and Lemma \ref{convergence lem}.
Consequently, $u$ is a bounded solution to \eqref{eq:heatLambda}.

The uniqueness of bounded solutions $u$ under the condition that $u_0:\Lambda\rightarrow \R$ is bounded
follows from \cite[Theorem 4.2]{Dodziuk}. We remark that the uniqueness of bounded solutions follows also from \cite[Theorem 7.2]{KLW} since $L:=-2n\Delta_\Lambda$ satisfies {\rm (i.a)} of \cite[Theorem 7.2]{KLW},
i.e., $(L+\alpha)1=\alpha$ for all $\alpha>0$.
\end{proof}

\begin{remark}[The Dedekind eta function]\label{heat eq and eta}
By Theorems \ref{thm:I-E:char} and \ref{th:heatLambda}, 
various theta functions can be obtained as a limit value at the position $x=0$ of heat equations on lattices. 
For example,
if $L\in \Z_{\ge1}$ is relatively prime to $12$, we have
	$$
	\left(\dfrac{12}{L}\right)\sum_{k\in\Z}\left(\dfrac{12}{Lk}\right)
	Le^{-6L^2t}I_{Lk}\left(6L^2t\right)
	=\dfrac{1}{\sqrt{12}}
	\sideset{}{'}\sum_{j=-6 L}^{6 L}
	\left(\dfrac{12}{j}\right)\exp\left(
	6L^2t(\cos \tfrac{\pi j }{6L}-1)
	\right)
	$$
	by \eqref{chi-I}.
	Since $\left(\tfrac{12}{-k}\right)=\left(\tfrac{12}{k}\right)$, 
	the left-hand side of the equality above is rewritten as
\begin{align*} \sum_{k\in\Z}\left(\dfrac{12}{-Lk}\right)
		Le^{-6L^2t}I_{Lk}\left(6L^2t\right)
		&=
		\sum_{k\in L\Z}\left(\dfrac{12}{0-k}\right)
		Le^{-6L^2t}I_{k}\left(6L^2t\right)
		\\
		&=
		L \cdot (\psi_{12,L} * \K_{\Z,6L^2t})(0),
		\end{align*}
		where $\chi_{12}(\cdot):=\left(\tfrac{12}{\cdot}\right)$, $\psi_{12, L}=\chi_{12}\1_{L\ZZ}$ and $\1_{L\ZZ}$ is the characteristic function of $L\ZZ$ on $\ZZ$.
	By Theorem \ref{th:heatLambda}, 
	the function $u(x,t):=L \cdot\chi_{12}(L)\cdot (\psi_{12,L} * \K_{\Z,t})(x)$ is a solution of the following heat equation on the lattice $\Z${\rm:}
	\begin{align}\label{eq:heatZ^n-eta}
		\begin{cases}
			\partial_t u(x,t) = \Delta_{\Z} u(x,t), & \qquad (x, t) \in \Z \times \R_{\ge0}, \\
			u(x,0) = u_0(x), & \qquad x \in \Z.
		\end{cases}
	\end{align}
	The solution of \eqref{eq:heatZ^n-eta} is also expressed as 
	$u(x,t)=\exp(t\Delta_{\Z})u(x,0)$.
	The value $L \cdot \chi_{12}(L)\cdot (\psi_{12,L} * \K_{\Z,6L^2t})(0)$ of the solution at the position $x=0$ and at time $6L^2t$ converges to 
	$\tfrac{2}{\sqrt{12\pi t}}\eta(-\tfrac{1}{\ii\pi t})=\frac{1}{\sqrt{3}}\eta(\ii\pi t)$ as $L\to\infty$
by Example \ref{ex:eta}.
In other words, we have
$$\lim_{L\to\infty}
	[\exp(6L^2t\Delta_{\Z})\,L\chi_{12}(L)\,\psi_{12,L}](0)
	=\frac{2}{\sqrt{12\pi t}}\eta(-\tfrac{1}{\ii\pi t})=\frac{1}{\sqrt{3}}\eta(\ii \pi t).
	$$
	This means that, for the initial condition $u(x,0)=L\cdot\chi_{12}(L)\psi_{12,L}(x)$,  the temperature at $x=0$ of 
	the space $\Z$ and at time $6L^2t$ converges to
	the Dedekind eta function as $L\to\infty$.  
\end{remark}

As an application of Theorem \ref{thm:code-main} and Theorem \ref{th:heat:Z^n}, 
we can consider the heat equation on $\Z^n$ with the initial condition assigned by a linear code 
and we obtain Theorem \ref{thm:heateq-code} as follows.
\begin{proof}[Proof of Theorem 1.3]
Note 
\begin{align*}
(u_0*\K_{\Z^n, t})(x)
=\sum_{\gamma \in \Z^n} \1_{\rho^{-1}(C)}(x-y) \K_{\Z^n, t}(y)
=\sum_{\substack{y \in \Z^n \\ x-y \in \rho^{-1}(C)}} \K_{\Z^n, t}(y).
\end{align*}
Since $x-y \in \rho^{-1}(C)$ is equivalent to 
$y \in x+\rho^{-1}(C)$, we have
\begin{align*}
(u_0*\K_{\Z^n,t})(x)
=\sum_{y \in x+\rho^{-1}(C)} \K_{\Z^n,t}(y)
=\sum_{y \in x+\rho^{-1}(C)}e^{-t}\prod_{j=1}^nI_{y_j}(\tfrac{t}{n}).
\end{align*}
Hence \eqref{solution} follows from \eqref{x-MacWilliams-I} in Theorem \ref{thm:code-main}.
Moreover, $\rho(x)\in C$ implies $\rho(x)+C=C$ and hence we obtain
\eqref{solution-c} by \eqref{diag-MacWilliams-I} for $x=0$ in Theorem \ref{thm:code-main}.
This completes the proof of Theorem \ref{thm:heateq-code}.
\end{proof}

\section*{Acknowledgements}
The authors would like to thank a referee for informing them of the paper \cite{XZZ}.
Takehiro Hasegawa was supported by JSPS KAKENHI (grant number JP22K03246).
Hayato Saigo was supported by JSPS KAKENHI (grant number JP22K03405).
Seiken Saito was supported by JSPS KAKENHI (grant number JP22K03405) and Research Origin for Dressed Photon. 
Shingo Sugiyama was supported by JSPS KAKENHI (grant number JP20K14298).


\end{document}